\documentclass[preprint,18pt]{elsarticle}

\usepackage{amssymb}
\usepackage{tabularx,booktabs,caption}
\newcolumntype{Z}{>{\raggedright}X}

\usepackage{multicol}
\usepackage{float}
\usepackage{xcolor}
\usepackage{amssymb}
\usepackage{bbm}
\usepackage{times}
\usepackage{soul}
\usepackage{url}
\usepackage[utf8]{inputenc}
\usepackage{graphicx}
\usepackage{hyperref}

\usepackage{amsmath}
\usepackage{amsthm}
\usepackage{booktabs}
\usepackage{algorithm}
\usepackage{algcompatible}
\usepackage{algpseudocode}
\usepackage{babel} 
\usepackage{apalike}
\usepackage{mathtools}
\usepackage{amsmath}
\usepackage{subcaption}
\DeclareMathOperator*{\argmax}{arg\,max}

\newtheorem{theorem}{Theorem}
\newtheorem{definition}{Definition}
  {
      \theoremstyle{plain}
      \newtheorem{assumption}{Assumption}
  }
  {
      \theoremstyle{plain}
      \newtheorem{Proposition}{Proposition}
  }

\begin{document}

\begin{frontmatter}



\title{Tree-based Focused Web Crawling \\ with Reinforcement Learning}


\affiliation[inst1]{organization={School of Electrical and Computer Engineering, National Technical University of Athens},
            addressline={9 Iroon Polytechniou}, 
            city={Athens},
            postcode={15780}, 
            % state={State One},
            country={Greece}}

\affiliation[inst2]{organization={Archimedes/Athena RC},%Department and Organization
            % addressline={Address One}, 
            % city={Athens},
            % postcode={00000}, 
            % state={State One},
            country={Greece}}

\affiliation[inst3]{organization={Institute of Informatics and Telecommunications, NCSR "Demokritos"},%Department and Organization
            % addressline={Address One}, 
            % city={Athens},
            % postcode={00000}, 
            % state={State One},
            country={Greece}}

\affiliation[inst4]{organization={SciFY P.N.P.C.},%Department and Organization
            % addressline={Address One}, 
            % city={Athens},
            % postcode={00000}, 
            % state={State One},
            country={Greece}}
            
\affiliation[inst5]{organization={Max Planck Institute for Software Systems},%Department and Organization
            % addressline={Address One}, 
            % city={Athens},
            % postcode={00000}, 
            % state={State One},
            country={Germany}}

\author[inst1,inst2]{Andreas Kontogiannis\corref{cor1}}
\ead{a.kontogiannis@athenarc.gr}

\author[inst1,inst3]{Dimitrios Kelesis}
\ead{dkelesis@iit.demokritos.gr}

\author[inst2,inst5]{Vasilis Pollatos}
\ead{v.pollatos@athenarc.gr}

\author[inst3,inst4]{George Giannakopoulos}
\ead{ggianna@iit.demokritos.gr}

\author[inst3]{Georgios Paliouras}
\ead{paliourg@iit.demokritos.gr}

\cortext[cor1]{Corresponding author}
% \fntext[fn1]{Postal Address: 9, Iroon Polytechniou St, Athens, 15780}

\begin{abstract}
A focused crawler aims at discovering as many web pages and web sites relevant to a target topic as possible, while avoiding irrelevant ones. Reinforcement Learning (RL) has been a promising direction for optimizing focused crawling, because RL can naturally optimize the long-term profit of discovering relevant web locations within the context of a reward. In this paper, we propose TRES, a novel RL-empowered framework for focused crawling that aims at maximizing both the number of relevant web pages (aka \textit{harvest rate}) and the number of relevant web sites (\textit{domains}). We model the focused crawling problem as a novel Markov Decision Process (MDP), which the RL agent aims to solve by determining an optimal crawling strategy. To overcome the computational infeasibility of exhaustively searching for the best action at each time step, we propose Tree-Frontier, a provably efficient tree-based sampling algorithm that adaptively discretizes the large state and action spaces and evaluates only a few representative actions. Experimentally, utilizing online real-world data, we show that TRES significantly outperforms and Pareto-dominates state-of-the-art methods in terms of harvest rate and the number of retrieved relevant domains, while it provably reduces by orders of magnitude the number of URLs needed to be evaluated at each crawling step.
\end{abstract}


% \begin{highlights}
% \end{highlights}


\end{frontmatter}


\section{Introduction}
Focused web crawlers \cite{CHAKRABARTI19991623} are intelligent agents which, given a topic of interest, find a strategy to search for relevant web pages. Due to the information abundancy on the Web, focused web crawling is a tool in high-demand by today's web services and users for retrieving large amount of information, relevant to the topic of interest. Typically, a focused crawler starts from an initial set of highly relevant data, in the form of URLs, called \emph{seeds}. The crawler utilizes the seeds to bootstrap the collection process, extracting their outgoing links, or \emph{outlink URLs}, and storing them in a URL collection, called the crawl \textit{frontier} \cite{10.1145/988672.988714}. The crawler selects the most promising URL from the frontier (based on some evaluation criterion), visits the corresponding web page and classifies it as relevant or irrelevant according to the user's interests. Then, recursively, the outlinks of the visited page are stored in the frontier. Based on relevance feedback by the user, the evaluation function that scores the frontier's URLs is adjusted. In a typical focused crawling setting, the total number of retrieved URLs is user-defined and the whole process terminates when that number of retrieved URLs is reached. A common goal in focused crawling, as in \cite{10.1145/1242572.1242632, 10.1007/s11280-015-0331-7, patil2016implementation, 10.1145/3308558.3313709, neelakandan2022automated, zhang2021dsdd, liu2022applying}, is maximizing: (a) the percentage of relevant retrieved web pages (often known as the \textit{harvest rate}, or \textit{accuracy}) and (b) the number of retrieved relevant different web sites (which we denote by \textit{domains}). 

Many approaches, such as \cite{6165295, RAJIV20213, SALEH2017181, Elaraby2019ANA, shrivastava2022efficient}, use classifiers to determine the crawling strategy by estimating whether the web pages in the frontier are relevant to the target topic. However, such approaches have limitations, since they are mostly capable  of finding relevant web pages within immediate crawling steps and fail to discover relevant pages that require the selection of possibly not relevant URLs from the frontier to reach them. Therefore, the crawler's effectiveness highly depends on the selection of seeds, which should extract a large number of outlink URLs into the frontier from different regions of the Web. To deal with such problems, approaches, such as \cite{10.1007/s11280-015-0331-7, 10.1145/3308558.3313709, zhang2021dsdd}, aim to enrich focused crawling with other search operations, including backward crawling, seed identification and/or keyword discovery through the use of search engines. However, such approaches may be impractical for an average user, since they require paid APIs. Thus, in this paper we focus on internally informed focused crawling, where no external sources of information are utilized.

% Such classifiers can be used to (a) identify the relevant web pages (in the form of URLs) in the frontier, assigning them high evaluation values, which imply high crawling priorities, (b) give relevance feedback to the strategy of the crawler once a web page has been visited and (c) evaluate the relevance of the crawled data at the end of the process. 

Reinforcement Learning (RL) has been a promising direction for optimizing focused crawling \cite{Rennie99efficientweb}, since RL can naturally optimize the long-term profit of discovering relevant web locations within the context of a reward signal. In particular, RL paramete- rizes the crawler agent's strategy (in the form of a learnable policy) to effectively \textit{explore} different but interesting locations of the Web, and also explicitly quantifies the value (i.e. the expected return) of the available actions (e.g. URLs in frontier) at each decision making (crawling) step, without the need of an external heuristic, and possibly handcrafted, criterion. In this way, RL crawling policies should be able to identify \textit{hubs} \cite{ding2004link} (web pages connected with relevant URLs on the Web) possibly reachable by following a sequence of not relevant URLs. That is, in cases where no high relevance URLs exist in the frontier, a good RL strategy may prioritize to reach hubs for retrieving highly relevant web pages at a later crawling step.

Although RL seems reasonable to use for optimizing focused crawling, most state-of-the-art focused crawlers do not exploit RL, due to some current challenges. Since the frontier size grows significantly in time, it is computationally infeasible in practice for the RL crawler agent to select the best URL (possibly associated with an action in the RL context) from the frontier. In particular, the brute force method (denoted by \textit{synchronous method} \cite{10.1007/978-3-319-91662-0_20}) for scoring all URLs in the frontier at each time step is shown to require an infeasible computational cost for long crawling sessions (see Section 2). On the other hand, RL approaches, such as \cite{Rennie99efficientweb, 10.1007/978-3-319-91662-0_20}, adopt sub-optimal heuristic methods, in order to approximate the synchronous method and optimize focus- ed crawling, which may harm performance dramatically (see Section 3). In particular, such methods lack substantial informativeness in their state representation, in their relevance estimation methodology and in the selection of candidate actions. Furthermo- re, such approaches  formulate the focused crawling problem as a Markov Decision Process (MDP), to solve it with RL, in a somewhat unnatural way that either forces the crawler agent to follow a single path of web pages or allows it to deviate from the successive order of state transitions.

In this paper, we deal with the above challenges and aim to establish RL as an efficient state-of-the-art method for focused crawling towards maximizing both the harvest rate and the number of retrieved domains, without requiring any access to possibly paid APIs. After formulating the problem as a mathematically sound Markov Decision Process (MDP), we face the problem of searching for optimal actions in intractably large state and action spaces and propose an efficient tree-based sampling algorithm. This algorithm adaptively compresses the search space and at each crawling step it selects the best-estimated action (associated with a URL to fetch next) over only a few representative actions. Our main contributions are the following:

\begin{itemize}
\item{} We propose a novel RL-based focused crawling framework called \textit{TRES (\textbf{T}ree \textbf{RE}inforcement \textbf{S}pider)}. The framework adopts a \textit{novel, mathematically sound Markov Decision Process (MDP) formulation} for focused crawling. 
\item We introduce the \textit{Tree-Frontier} sampling algorithm, which might be of independ- ent interest outside the context of focused crawling. The algorithm provides efficient and effective representation of the crawl frontier, through adaptive discr- etization of the large state and action spaces. Moreover, we provide mathematical analysis and intuition of the algorithm explaining its logic and interpreting the good experimental performance. 
% \item{} We propose a novel approach of utilizing topic-related keywords, based on automatic suggestions of such keywords - we call this process \emph{keyword set expansion} - to form a richer keyword collection. The approach utilizes keyword similarities to construct vector representations of candidate URLs summarizing semantic and statistical properties of the URLs.
\item We show that \textit{TRES} \textit{outperforms and Pareto-dominates other state-of-the-art} focused crawlers in three test settings, in terms of both harvest rate and the number of retrieved domains.
% , utilizing {no more than a single seed URL}.
\item We show that the use of Tree-Frontier \textit{provably improves the time complexity} of the impractical brute force method (denoted by synchronous method \cite{10.1007/978-3-319-91662-0_20}) for estimating all available actions at a given state.
\end{itemize}

The remainder of the paper is organized as follows. In Section 2, we present the problem definition and the proposed framework, along with its mathematical analysis and intuition. The experimental details, the comparisons with baseline and state-of-the-art methods and the ablation study of the TRES are presented in Section 3. In Section 4, we position this paper with respect to the related work.  Section 5 summarises the conclusions of this paper and suggests future research directions. 




\section{Problem and Method Overview} 
In this section, we formulate the problem setting that we aim to solve. First, we describe the keyword expansion problem that we deal with in order to discover keywords relevant to the topic of interest that will next guide the crawling process. Second, we provide the MDP formulation for the focused crawling problem and present the overview of the TRES\footnote{Unlike most focused crawlers in the literature, the implementation code of TRES is publicly available: \href{https://github.com/ddaedalus/tres}{https://github.com/ddaedalus/tres}.} framework.

    \subsection{Problem Definition} 
        \subsubsection{Keyword Expansion Problem}
        Our goal is to find a set of keywords that is sufficient to identify  any document, belonging to a given topic. In order to do that, similar to \cite{10.1007/s11280-015-0349-x}, we start with a small number of keywords and adopt a keyword expansion strategy in order to identify as many keywords as possible. Formally,
        let $D_C$ be the set of all web documents belonging to a target topic $C$. Also, let $K_C = \{k_1, ..., k_{N_K}\}$  be the finite set of all keywords that describe C. Given an initial small set of keywords, $KS = \{k_1, ..., k_{N_{KS}}\}$, where $KS \subset K_C$, and a corpus of text documents $D_{tr} = \{t_1, ..., t_N\}$ where each $t_i = \{w^i_1, ..., w^i_{n_i}\}$ contains $n_i$ candidate keywords, our goal is to expand $KS$ so that $KS = K_C$.

        \subsubsection{Focused Crawling Problem}
        Let $U_C$ be the finite set of URLs corresponding to all web pages that belong to the target topic $C$. Given $C$ and a small set of \textit{seeds} $U_S$, where $|U_S| << |U_C|$, our goal is to expand $U_S$ over a $T$-time step-attempt-process of discovering URLs belonging to $U_C \setminus U_S$. Since the crawling process is terminated after $T$ time steps, $U_S$ would be expanded by $T$ new URLs. We will refer to these $U_C \setminus U_S$ URLs discovered as \textit{relevant}.

        In the rest of this section, we describe our proposed TRES framework. First, we formulate the focused crawling problem as an MDP. Then, we present the keyword expansion strategy that our crawler adopts, in order to discover keywords that are relevant to the target topic and aim to enhance the focused crawling performance. 
        % Next, we describe KwBiLSTM, a deep neural network estimating the reward function of our RL setting. 
        Next, we present the proposed focused crawling algorithm, which utilizes RL, in order to find good policies for selecting URLs, and the Tree-Frontier algorithm for efficient sampling URLs from the frontier.

        \subsection{Modeling Focused Crawling as an MDP} \label{sec:MDPFormulation}
        % Similarly to \cite{10.1007/978-3-319-91662-0_20}, we utilize the crawler as an agent, which exists in an interactive environment providing states, actions and rewards. Its goal is to maximize the cumulative reward, which is equivalent to discovering as many relevant URLs as possible. The latter is exactly the definition of the focused crawling task. To model a focused crawling setting as an MDP, we first describe some specific concepts. 
        
        Similar to \cite{Rennie99efficientweb, 10.1007/978-3-319-91662-0_20}, we consider the focused crawler to be an agent, which operates in an interactive environment. Within such an environment, the agent changes states, takes actions and receives rewards, with the goal of maximizing the cumulative reward. In this context, focused crawling can be modeled as an MDP unifying the traditional use of the frontier of a focused crawler with the action set of an RL agent.
        
        \subsubsection{Preliminaries}
        
        Let $U$ be the finite set of all URLs on the Web. Let Page: $U \rightarrow P$ be the function that maps URLs to corresponding web pages in $P$. $P$ refers to the set of existing pages on the Web. Let $\text{Outlink}: U \rightarrow 2^U$ be the function that matches a given URL $u \in U$ to a set containing its outlink URLs. 
        
        Let $G$ be the directed graph of the {Web}. The node set of $G=(U,E)$ is equivalent to the URL set $U$. The edge set $E$ of $G$ is defined as follows: a node $u_i \in U$ has a directed edge to another node $u_j \in U$ if $u_j \in \text{Outlink}(u_i)$. We represent the corresponding edge with $(u_i,u_j) \in E$.
        
        In focused crawling, the crawler gradually traverses a subgraph of G. Each time it fetches a new URL, the corresponding node (associated with its web page) is added to the subgraph. Let $g_t$ be the directed subgraph of $G$ that the crawler has traversed until time step $t$. For $t=0$ it holds that $g_0=(U_S,\emptyset)$; i.e. the initial subgraph has its node set equal to the seed URL set $U_S$ and no edges. We denote the URL fetched at $t$ as $u^{(t)}$. Then, $g_{t+1}$ is $g_t$ expanded by node $u^{(t+1)}$ and edge $(u^{(k)},u^{(t+1)})$, where $0 \leq k < t+1$. Therefore, $g_t$ contains all $u^{(k)}$ nodes, for $0 \leq k \leq t$. Note that for each time step $t$, $g_t$ is a forest of $|U_S|$ ordered trees.
        
        Also, let \textit{closure}, or \textit{crawled set} as defined in \cite{10.1145/2631775.2631795}, $C_t$ be the set of fetched URLs until time step $t$. Note that closure $C_t$ is equal to the nodes fetched at time step $t$. From the definition of closure, the crawler fetches a URL at most once, so a URL has a unique web path for the whole crawling process. Now, we introduce the definition of a web path, which will be necessary to define appropriately the frontier.
        
        % More formally:
        
        % \begin{definition}[Closure]
        % We define closure $C_t$ at time step $t$ as:
        % $C_t = \{u^{(k)} : 0 \leq k \leq t \}$; i.e. the set of nodes of the traversed subgraph $g_t$.
        % \end{definition}
        
        \noindent
        % Now, we introduce two definitions that are important for modeling the MDP. 
        
        % \begin{definition}[Parent URL]
        % If a URL $u_i$ has a directed edge $(u_iu_j)$ to a URL $u_j$ in $g_t$, then we call $u_i$ the parent URL of $u_j$ and we write $u_i = par(u_j)$. 
        % \end{definition}  

        \begin{definition}[Web Path]
        The web path of the URL $u^{(t)}$ fetched by the crawler at time step $t$, denoted by $path(t)$, is the path of $g_t$ from a given seed URL in $U_S$ to $u^{(t)}$.
        \end{definition}  
        
        % \begin{definition}[Web Path]
        % We define the web path $path(t)$ of the URL $u^{(t)}$ fetched by the crawler at time step $t$ as the longest path of $g_t$ that includes $u^{(t)}$. Practically, this is a path from a given seed URL in $U_S$ to $u^{(t)}$.
        % \end{definition}  

        \noindent
        Let \textit{frontier}, or \textit{crawl frontier}, $F_t$ be the set of all available outlink URLs that the crawler has not fetched until time step $t$, but were extracted by web pages fetched before $t$. More formally:
        
        \begin{definition}[Frontier]
        We define frontier $F_t$ at time step $t$, given $C_t$, as:
        \begin{align*}
        % \resizebox{1\linewidth}{!}{
        F_t = \left\{ (path(\tau), u): \tau \geq 0, u^{(\tau)} \in C_t, u \in \text{Outlink}(u^{(\tau)}), u \not \in C_t \right\} 
        % }
        \end{align*}
        \end{definition}  
        
        \noindent
        In the above definition, notice that we also include the corresponding web paths of those URLs, from which we extract the available URLs to fetch at time step $t$. This definition of frontier aims to allow the crawler agent to identify patterns from single web paths that possibly lead to hubs or relevant pages. Moreover, the above definition does not force the crawler agent to follow a single path of web pages, as all available web paths are included in the frontier at each time step.
        
        \subsubsection{MDP formulation}
        
        Based on the above, we formulate the crawling process an the MDP, represented by the tuple $(S, A, P_{t,a}, R)$. We denote a \textit{state} $s_t \in S$ at time step $t$ as the directed subgraph of $G$ that the crawler has traversed at $t$; i.e. $s_t=g_t$. We define actions ${A}_t \subset A$ at time step $t$ given $s_t$ (and thus $g_t$ and $C_t$), as the set of all available edges through which $g_t$ can be expanded by one new node; that is: 
        \begin{align}
        % \resizebox{1\linewidth}{!}{
            {A}_{t} = \left\{ (v,u) : v \in C_t, u \in \text{Outlink}(v), u \not \in C_t \right\}
        % } 
        \end{align}
        
        \noindent
        Also, we define $P_{t,a}(s,s') = \mathbb{P}[S_{t+1} = s' \  \vert \  s_t = s, a_t = a]$ as the transition probability of selecting the edge $a$ to expand state $s = g_t$ in order to transition to a state $s' = g_{t+1}$. In our setting, we utilize deterministic transitions. Furthermore, we define a binary reward function $R\colon S \times A \times S \rightarrow \mathbb{R}$ that equals $1$ if the web page of the URL $u^{(t+1)}$ crawled at time step $t+1$ is relevant and $0$ otherwise. For brevity, let $r_{t+1}$ be the reward observed at time step $t$, that is $r_{t+1} = R(s_t, a_t, s_{t+1})$. 
        
        % Notice that the number of different states is equal to the number of all different directed subgraphs of Web graph $G$, given the seed set $U_S$.
        The initial state $S_0 = g_0$ consists of the seed URLs. A final state is not naturally defined since crawling is an infinite-horizon problem. Assuming that the crawling process stops after $T$ time steps, $g_T$ is the final state. 
        
        From all the above, unlike previous approaches that regard action selection to be related to URL relevance estimation, to our knowledge, our approach is the first to both expand multiple web paths of the Web graph and preserve the successive order of state transitions by always selecting actions that are explicitly available in the current state. Also, observe that in our MDP, given a time step $t$, selecting a $(path(\tau), u)$ from $F_t$ is equivalent to selecting action $(u^{(\tau)},u)$ from ${A}_t$. Recall that ${A}_t$ is the set of all available actions at time step $t$. Now we can match these actions to the elements of frontier $F_t$ with URLs not in closure $C_t$. Therefore, we can relate the traditional use of the frontier to the action set of an RL focused crawler agent for all time steps.
        
        The agent aims to maximize the expected discounted cumulative reward, which is formulated as $G_t = \sum_{\tau=t}^{\infty}\gamma^{\tau-t}r_{\tau}$.
        Here, $\gamma \in [0,1]$ is a discount factor which is used for trading-off the importance of immediate and future rewards. Under a stochastic policy $\pi$, the Q-function of a state-action pair $(s,a)$ is defined as follows
        
        \begin{equation}
        Q^{\pi}(s,a) = \mathbb{E}\left[G_t \mid s_t=s, a_t=a, \pi \right] \nonumber
        \end{equation}
        
        \noindent
        which can also be computed recursively with bootstrapping:
        
        \begin{equation}
        Q^{\pi}(s,a) = \mathbb{E}\left[r_t + \gamma \mathbb{E}_{a' \sim \pi(s')}[Q^{\pi}(s',a')] \mid s_t=s, a_t=a, \pi \right] \text{.} \nonumber
        \end{equation}
        
        \noindent
        The Q-function measures the value
        of choosing a particular action when the agent is in a given state. We define the optimal policy $\pi^*$ under which we receive the optimal $Q^*(s,a) = \max_{\pi}Q^{\pi}(s,a)$. At time step $t$, given a state $s_t$, under the optimal policy $\pi^*$, the agent selects action $a_t = \argmax_{a' \in {A}_t} Q^*(s_t,a')$. Therefore, it follows that the optimal Q-function satisfies the Bellman equation:
        
        \begin{equation}
        Q^*(s,a) = \mathbb{E}\left[r_t + \gamma{} \max_{a' \in {A}_t} Q^*(s',a') \mid s_t=s, a_t=a, \pi \right]\nonumber
        \end{equation} 

    In practice, to decide whether a retrieved web page is relevant or not, we train a classifier that will constitute the reward function of the proposed MDP. To this aim, in practice, we utilize a bidirectional LSTM (BiLSTM) \cite{10.1007/11550907_126}, a model that has shown remarkable performance in web page classification and relevance estimation for focused crawling when trained on small (or medium) datasets \cite{shrivastava2022efficient, neelakandan2022automated, axiotis2021personalized}. We note that our choice of the traditional BiLSTM over a more complex Transformer-based network is  because the latter is more data-hungry and in some early experiments with small datasets it was less robust.

    \subsection{Discovering Keywords Relevant to the Target Topic}
    Keywords and seeds often play the role of a crawler's only prior knowledge about the target topic \cite{Menczer99adaptiveretrieval, 10.1145/1031114.1031117, zhang2021dsdd}. A collection of keywords that are irrelevant to the target topic can lead the crawler to follow URLs that diverge from the topic of interest. In this paper, we assume that an initial small keyword set $KS = \{k_1, ..., k_{N_K}\}$ is given as input, containing keywords that are all highly related to the target topic $C$. We also use a keyword expansion method that discovers new keywords from a corpus of text documents, $D_{tr}$. 

    \begin{algorithm}[t]
    \caption{Keyword Expansion Strategy}\label{alg:1}
        \hspace*{\algorithmicindent} \textbf{Input:} Initial keyword set $KS$ of size $N_K$, a corpus of text documents $d_{tr}$ \\  
        \hspace*{\algorithmicindent} \textbf{Output:} the expanded keyword set $K$ 
    \begin{algorithmic}[1]
        \State Initialize empty keyword set $K' = \{\}$
        \State Create a set $W$ of all words in $D_{tr}$
        \For {$w \in W$} 
            \If{$\frac{1}{N_K}\sum\limits_{k\in KS}cos(w,k) \geq b$} \Comment{Equation (\ref{eq:2})} 
                \State append $w$ in $K'$
            \Else 
                \State continue
            \EndIf
        \EndFor
        \State $K = KS \cup K'$
    \end{algorithmic}
    \end{algorithm}

    We consider all words of the texts from $D_{tr}$ to be candidate keywords of the target topic. Similarly to \cite{neelakandan2022automated, 10.1007/978-3-319-91662-0_20, dhanith2021word, liu2022applying, wang2023hunger}, we use the cosine similarity of word vectors to measure the similarity with the topic of interest. In particular, in our setting each candidate word is represented as a vector in $\mathbbm{R}^N$, using a pretrained word2vec\footnote{\href{http://vectors.nlpl.eu/explore/embeddings/en/}{http://vectors.nlpl.eu/explore/embeddings/en/}} \cite{DBLP:journals/corr/abs-1301-3781} model, that has been trained on the whole Wikipedia corpus. To decide whether a candidate word $w$ is a new keyword, we measure a simple semantic score $CS(w)$, as the average cosine similarity of $w$ with each of the existing keywords in the collection $KS$. 
    
    If the semantic score exceeds a given threshold $b$, $CS(w_i) > b$, then $w_i$ is regarded as a keyword and is stored in another keyword set $K'$. At the end of this process, the selected keywords are combined with those of the initial set: $K = KS \cup K'$. In this paper, empirically we set the threshold $b$ to the average cosine similarity of all $k \in KS$. More formally:

    \begin{align}
        b = \frac{1}{N_K(N_K-1)} \sum\limits_{i=1}^{N_K}{\sum\limits_{j \neq i}^{N_K}}{cos(k_i, k_j)} \label{eq:2}
    \end{align}    

    \noindent
    We present our keyword expansion strategy in Algorithm \ref{alg:1}. 
    Note that the strategy tends to be highly selective, since a new keyword must have a greater $CS$ score than some keywords in $KS$.

    \subsection{Optimizing Focused Crawling with Reinforcement Learning}
    In this subsection, we describe our RL focused crawler agent and the process of learning efficiently good Q-value approximations for selecting available URLs from the frontier. Our approach is divided into three parts: (a) representing the states and the actions, (b) efficiently sampling URLs from the frontier, and (c) training the agent with an RL algorithm. 
    
    \subsubsection{State-Action Representation}
    Similar to \cite{10.1007/978-3-319-91662-0_20}, we propose a \textit{shared state-action representation}. Let $\pmb{x}(s_,a_t)$ be the shared representation vector for a given state $s_t$ and an action $a_t$ at time step $t$. Recall that an available action $a$ (at time step $t$) is related to a URL $u$ in the corresponding frontier, $F_t$, such that $a = (u^{(\tau)}, u)$, where $\tau$ is the time step when the crawler fetched URL $u^{(\tau)}$ and $u$ is one of the outlinks of $u^{(\tau)}$, that is $u \in \text{Outlink}(u^{(\tau)})$. 

    To represent state $s_t$ in $\pmb{x}(s_t, a_t)$, given the action $a_t$ and the web page $u^{(\tau)}$ that was fetched by the crawler at time step $\tau$ and that extracted the URL of this action, we simply aggregate information of $path(\tau)$; i.e. the web path leading to URL $u^{(\tau)}$. Specifically, similarly to \cite{10.1007/978-3-319-91662-0_20}, we use the following scalar state features: the reward received at $\tau$, the inverse of the distance of $u^{(\tau)}$ from the closest relevant node in $path(\tau)$ and the relevance ratio of $path(\tau)$. Intuitively, these features describe the relevance of the web path related to the given outlink URL. 
        
    To represent $a_t$ in $\pmb{x}(s_t,a_t)$, our method differs from \cite{10.1007/978-3-319-91662-0_20} and uses the following scalar features: the existence of keywords in the URL text, the existence of keywords in the anchor text of the source page and a probability estimation of its relevance given by the classifier (e.g. a BiLSTM), based on the outlink URL's anchor text. 
    % A web page belonging to a relevant web site is likely to lead to other relevant ones, even if this particular page is not relevant. The web pages with this property are known as hubs \cite{Rennie99efficientweb}. 
    We also introduce two scalar \textit{web domain} action features that describe the expected relevance of a domain (web site) until the current time step, which are the following: (a) the ratio of the relevant web pages of the URL's domain found until time step $t$ and (b) the unknown domain relevance which is set to $0.5$ if the specific web site (i.e. the domain portion of the URL) of the web page has not been visited before, otherwise $1$. These features are used on the assumption that the crawler is more likely to fetch more relevant web pages by trying to avoid less relevant or unknown domains.

    \begin{figure}[t]
        \centering
        \includegraphics[scale=0.15]{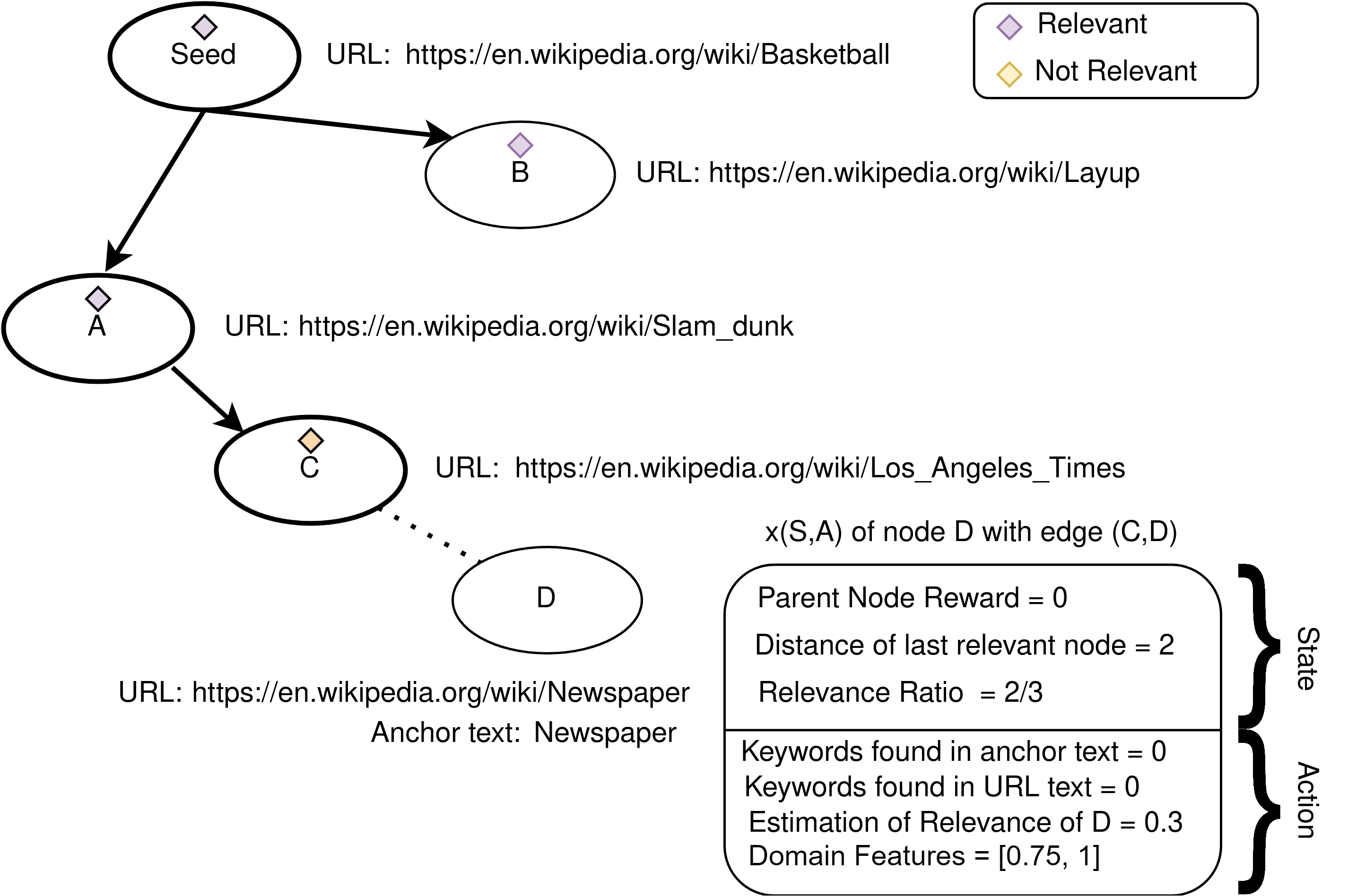}
        \caption{A simple state and action example: The illustrated state contains the nodes Seed, A, B and C, while the edge (C, D) is an action for this state.}
        \label{fig:x_repr}
    \end{figure}

    An example of a state-action pair is depicted in Figure \ref{fig:x_repr}. In this example, we examine crawling in the Sports domain. Starting from one seed URL, the crawler has fetched two relevant web pages (A, C) and one irrelevant (B). The current state is the Web subgraph containing the seed and the web pages A, B and C. The crawler examines the candidate action of fetching web page D. We also present the state-action representation of the current state and the candidate action. In the state representation, the parent node reward is zero (C is not relevant), the distance of last relevant node (A) is 2 and the relevance ratio of the web path (seed, A, C) is $2/3$ (the relevant nodes in the web path are the seed and A). In the action representation, no keywords are found in either the anchor text or the URL text of D and this accounts for the first two zero entries. Moreover, the estimation of relevance produced by the classifier is $0.3$. The web domain features are $3/4$ and $1$. The former is due to the fact that in this domain (en.wikipedia.org) the crawler has identified 3 relevant URLs and $1$ not relevant URL. The latter, which equals $1$, implies that this domain (en.wikipedia.org) has been fetched at least once. Otherwise, the second web domain feature would have been $0.5$.
    
    \subsubsection{Training with Reinforcement Learning}
    In our setting, the focused crawler is an RL agent operating in the MDP described in Section \ref{sec:MDPFormulation}. Here, we utilize the Double Deep Q-Network (DDQN) \cite{10.5555/3016100.3016191} agent, given its training stability and the recent success of deep Q-learning in web crawling \cite{wang2019toward, avrachenkov2021deep}; though other model-free RL agents that estimate Q-values on a state-action input could be also used in its place. Let $\hat{Q}$ be the neural network estimation of the state-action value function of the agent's policy. Specifically, let $\pmb{\theta}$ be the parameters of the online Q-Network and $\pmb{\theta}^{-}$ the parameters of the target Q-Network of DDQN. Then, given a record $(\pmb{x}(s_t,a_t), r_{t+1}, s_{t+1})$ from Experience Replay $B$ \cite{Mnih2015HumanlevelCT}, that is utilized in Deep Q-learning settings for sampling past state transitions, the DDQN target \cite{10.5555/3016100.3016191} can be written as
    
    \begin{equation}
        y_t = r_{t+1} + \gamma \hat{Q}_{\pmb{\theta}^{-}}\Big(\pmb{{x}}\big(s_{t+1},\argmax_{a' \in \tilde{A}_{t+1}} \hat{Q}_{\pmb{\theta}}(\pmb{{x}}(s_{t+1}, a'))\big)\Big) \nonumber
    \end{equation}   
    
    \noindent
    where $\tilde{A}_{t+1}$ is the union of the set of actions extracted at time step $t+1$ and the action returned by the tree frontier at time step $t+1$. More formally, we can compute $\tilde{A}_{t+1}$ as follows: 
    
    $$ \tilde{A}_{t+1} = \{ (u^{(t+1)}, u') : u' = Outlink(u^{(t+1)}) \} \cup \{\text{sampling}(F_{t+1},\\D_F(t+1))\} $$ 

    \noindent
    where $\text{sampling}(F_{t+1},D_F(t+1))$ is described in Algorithm \ref{alg:sampling} and  $u^{(t+1)}$ denotes the web page fetched at time step $t+1$. $\tilde{A}_{t+1}$ is used as an approximation of ${A}_{t+1}$, the full set of available actions at time step $t+1$, because the latter is intractably large for exhaustive search. As we will see in Section 2.6, $\text{Tree-Frontier sampling}$ is designed to approximate $\argmax_{a \in {A}_{t+1}} \hat{Q}_{\pmb{\theta}}(\pmb{{x}}(s_{t+1}, a'))$, an operator appearing both in the training of (standard) DDQN and the policy execution. We also evaluate the DDQN target based on the actions of the immediate outlink URLs of the currently fetched web page to improve the approximation of $\argmax$. This set of actions is a good candidate set for the $\argmax$, assuming that the currently fetched web page is relevant (which should happen with high probability if the crawler works well) and may probably lead to other relevant web pages. We note that although the latter candidate set may lack informativeness, in \cite{10.1007/978-3-319-91662-0_20} it was the only one used to estimate the $\argmax$ (in the asynchronous method).  
    
    Training starts from initial state $s_{-1}$, which represents the empty graph. Then, we initialize $B$ leveraging the experience from seeds $U_S$, using zero state features in $\pmb{x}(s_{-1},a)$ and positive rewards. Note that for each seed URL the action features in $\pmb{x}(s_{-1},a)$ are related to selecting this URL as the next action. This initialisation can speed up training, in cases where a large number of seeds are given as input and/or positive rewards are sparse in agent's exploration time. Next, we minimize $ {\mathbb{E}}_{U(B)}\left[y_t - \hat{Q}_{\pmb{\theta}}\left(\pmb{{x}}(s_t, a_t)\right) \right]^2$, with respect to $\pmb{\theta}$ by performing gradient descent steps on mini-batches of size $B$. We note that the function approximators used to produce $\hat{Q}$ for both the online and target Q-Networks are standard Multilayer Perceptrons (MLPs) with two hidden layers. 
    
    \subsection{Synchronous Frontier Method}
    The only part that we do not have described yet is  which action $a \in {A}_t$ is selected at a given time step $t$. In a common Deep Q-learning setting, an $\epsilon$-greedy policy is often used \cite{Mnih2015HumanlevelCT, 10.5555/3016100.3016191, kontogiannis2023xdqn}, which results in calculating the estimated Q-values of all actions of a given state to find $\argmax_{a \in {A}_{t}} \hat{Q}_{\pmb{\theta}}(\pmb{{x}}(s_{t}, a))$. 
    
    In the focused crawling setting, it is common to implement the frontier as a priority queue, with the priority values being the estimated Q-values of \textit{all} actions. In that case, for each time step, updating $\pmb{\theta}$ uses the \textit{synchronous method} \cite{10.1007/978-3-319-91662-0_20}. Let $P_{DDQN}$ be the time complexity of a single DDQN prediction and $d(t)$ be the number of new URLs inserted in the frontier at time step $t$. In the following theorem, we present the time complexity of the synchronous method (we note that the theoretical computational cost of the method was not discussed in the original work \cite{10.1007/978-3-319-91662-0_20}).
    % = \left|u: u \in Outlink(u^{(t)}), u \not\in C_t\right|$ 

    \begin{theorem}[Time Complexity of the Synchronous Method \cite{10.1007/978-3-319-91662-0_20}]\label{theorem:frontierSize}
        Assuming that for each time step the frontier size grows at least by $e_d$ and at most by $D$, let $S$ be the number of seed URLs. Then the synchronous method \cite{10.1007/978-3-319-91662-0_20} has an overall time complexity of at least $P_{DDQN} [(S e_d + \frac{e_d}{2}) T + \frac{e_d}{2} T^2)]$, and of at most $P_{DDQN} [ (S D + \frac{D}{2}) T + \frac{D}{2} T^2]$.
    \end{theorem}

    \begin{proof}
    Let $F_0$ be the frontier instance because of seeds. The time complexity of the frontier prediction through a synchronous update for the whole crawling process is equal to $P_{DDQN} \cdot F$, where $F = \sum\limits_{t=0}^{T}|F_{t}|$, $|F_t|$ is the frontier size at time step = $t$ and $P_{DDQN}$ is the time complexity of a single DDQN state-action value function prediction. Also, let $|F_0|$ be the frontier size because of the seed outlink URLs. Assumi- ng that for each time step the frontier size grows at least by $e_d$, we can lower-bound $|F_t|$:
    
    \begin{align}
        |F_t| & = |F_0| + \sum\limits_{\tau=1}^{t}{d(\tau)} \geq S e_d + \sum\limits_{\tau=1}^{t}{e_{d}}
        = S e_d + t e_{d} \nonumber 
    \end{align}
    
    \noindent
    Thus, we can lower-bound $F$:
        \begin{align}
        F = \sum\limits_{t=0}^{T}|F_{t}| \geq \sum\limits_{t=1}^{T} (|F_0| + t e_{d}) \geq (S e_d + \frac{e_{d}}{2}) T + \frac{e_{d}}{2} T^2 \nonumber
        \end{align}
    
    \noindent 
    Therefore, the time complexity of the frontier prediction (through synchronous update) for the whole crawling process is at least $P_{DDQN} ((s e_d + \frac{e_{d}}{2}) T + \frac{e_{d}}{2} T^2)$. Similarly, assuming that $d(t) \leq D$, where $D$ is the maximal number of outlink URLs from any web page on the Web (and thus practically $D$ and $T$ can be within the same order of magnitude), we can upper-bound $F$, as follows:
    $$F \leq (S D + \frac{D}{2}) T + \frac{D}{2} T^2$$

    \end{proof}

    \noindent
    Theorem \ref{theorem:frontierSize} indicates that the synchronous method is costly and, thus, impractical, due to the dependence on arbitrarily large $e_d$ and $D$. In practice, focused crawling has to deal with very large frontier sizes, and thus it is impossible to use the synchronous method for training and running an RL focused crawler. In our experiments (see Figure \ref{fig:frontier}), we will discuss the real-time huge growth of frontier that makes the synchronous update very difficult to run in practice. 
    
    \subsection{Updating and Sampling through Tree-Frontier}
    To reduce the time complexity of the  synchronous (brute force) method, we introdu- ce \textit{Tree-Frontier}, a two-fold variation of the CART decision tree algorithm, which we use to approximate the operator $\argmax_{a \in {A}_{t}} \hat{Q}_{\pmb{\theta}}(\pmb{{x}}(s_{t}, a))$,  appearing in the training of DDQN and the policy execution. In Tree-Frontier, frontier $F_t$ has now a decision tree representation, where each frontier element is assigned to a leaf. Decision trees have been used in RL to find a discretization of a large state space by recursively growing a state tree \cite{10.5555/295240.295802}. The set of tree leaves forms a partition of the initial state space. By splitting tree nodes, narrower convex regions are created, in which the agent behaves in a predictable way. In our case, each partition corresponds to a group of frontier samples of the form $\pmb{x}(s,a)$.  
    
    To achieve a meaningful partition, we borrow concepts from Explainable RL (XRL), where decision trees are often used to provide interpretable policy approximations \cite{Wu2018BeyondST, Liu2018TowardID, Bewley_Lawry_2021, vouros2022explainable, kontogiannis2023inherently}. Most importantly, Bewley and Lawry \cite{Bewley_Lawry_2021} proposed a binary decision tree providing MDP state abstraction through a linear combination of three impurity measures:  the action selected,
    the expected sum of rewards, and the estimated state temporal dynamics.
    As a result, their method identifies convex regions of the state space that are as pure as possible, in terms of the above criteria. Furthermore, to address the problem of allocating memory for huge continuous state spaces, Jiang et al. \cite{Jiang2021AnER} utilized a decision tree for learning the environment dynamics which also played the role of the experience replay buffer \cite{Mnih2015HumanlevelCT} of an RL agent for generating simulated samples. 
    
    Inspired by the above recent works, we are interested in discretizing simultane- ously both the state and the action spaces, which are defined in potentially continu- ous vector spaces, in an online manner such that this discretization facilitates optimal focused crawling. We distinguish two different kinds of samples in the frontier. For a given time step $t$, let \textit{experience samples} be those $\pmb{x}$ representation vectors that were selected at previous time steps $\tau < t$, and thus we have received their respective rewards $r_{\tau+1}$. Also, let \textit{frontier samples} be the $\pmb{x}$ representation vectors belonging to frontier $F_t$. Combining the above, let:
    
    $$D_E(t) = \{\left(\pmb{x}(s_{0},a_{0i}), 1\right) \colon 1\leq i \leq |U_S|\} \cup \{\left(\pmb{x}(s_{{\tau}},a_{{\tau}}), r_{{\tau+1}}\right) \colon \tau < t\}$$ 
    
    \noindent
    be the set of experience samples (including seeds) at time step $t$. Similarly, let $D_F(t) = \{\pmb{x} \colon \pmb{x} \in F_t \}$ be the set of frontier samples at time step $t$ given $F_t$.

\begin{algorithm}[t]
    \caption{Tree-Frontier sampling at time step  $t$}\label{alg:sampling}
    \begin{algorithmic}[1]
        \State \textbf{Input:} Tree-Frontier $F_t$, $D_F(t)$
        \State \textbf{Output:} a state-action $\pmb{x}$
        \State Select a representative frontier sample $\pmb{x}_i$ from each leaf through uniform sampling and remove it from the tree
        \If{\textit{exploration mode}:}  $\pmb{x} =$ $\text{UniformSampling}(\pmb{x}_i)$
        \Else: \,$\pmb{x} = \argmax_i \hat{Q}_{\pmb{\theta}}(\pmb{x}_i)$
        \EndIf
    \end{algorithmic}
    \end{algorithm}

    \begin{algorithm}[t]
    \caption{Tree-Frontier update at time step  $t$}\label{alg:2}
    \begin{algorithmic}[1]
        \State \textbf{Input:} Tree-Frontier $F_t$, $D_E(t)$, $D_F(t)$, new experience sample $e_{new} = \left(\pmb{x}(s_{t},a_{t}), r_{{t+1}}\right)$, new frontier samples $f_{new}$ 
        \State \textbf{Output:} a state-action $\pmb{x}$, $D_E(t+1)$, $D_F(t+1)$ and $F_{t+1}$

        \State Update $D_E(t+1) = D_E(t) \cup e_{new}$
        \State Update $F_{t+1}$ from $F_t$: check for split on leaf $P$, that contains $e_{new}$, using equations (\ref{split}) and (\ref{vr})
        \State Update $D_F(t+1) = D_F(t) \cup f_{new}$ by inserting each sample in $f_{new}$ to a leaf in $F_{t+1}$ following the tree rules
        \State$\pmb{x} =$ Tree-Frontier sampling($F_{t+1}$, $D_F(t+1)$) ; Algorithm \ref{alg:sampling}
    \end{algorithmic}
    \end{algorithm}    

    \begin{figure}[t]
        \centering
        \includegraphics[scale=0.6]{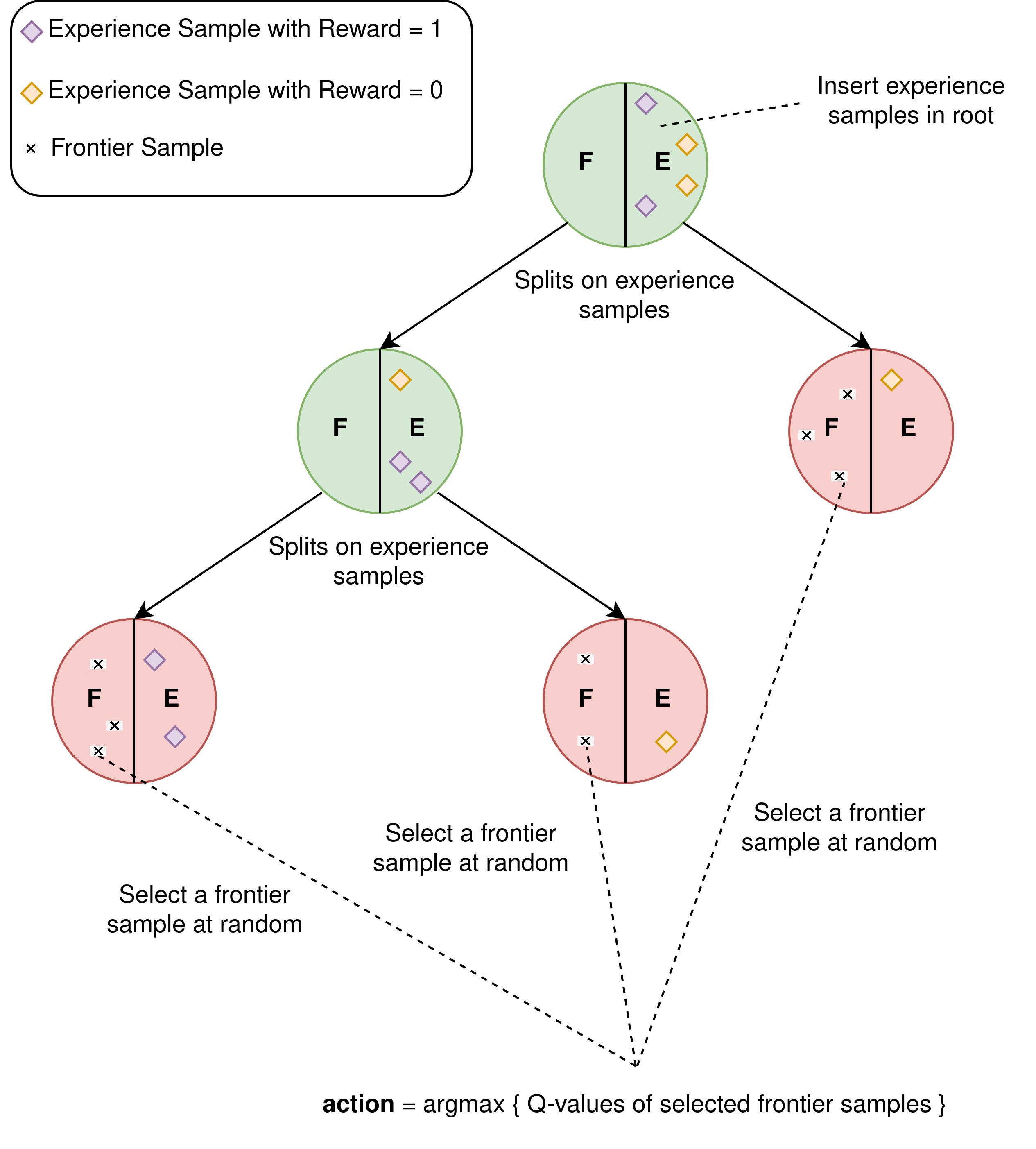}
        \caption{An example of Tree-Frontier sampling}
        \label{fig:tree_front}
    \end{figure}

    We propose a novel online binary decision tree algorithm, called Tree-Frontier, in order to efficiently represent and manipulate the frontier. As mentioned earlier, Tree-Frontier is two-fold; for a given time step $t$, based on the splitting rules, a tree node stores both a subset of $D_E(t)$ and a subset of $D_F(t)$. These sets are effectively the labeled (``training'') and unlabeled (``prediction'') ones. Tree-Frontier uses the experien- ce samples, $D_E(t)$, in order to split nodes, utilizing their rewards as target labels. For a given frontier sample (a state-action vector), Tree-Frontier predicts the agent's reward. Let $I_P$ be the subset of experience samples of $D_E(t)$ that belong to a tree leaf $P$. To split this leaf (parent), we seek binary partitions $I_P = \{I_L, I_R\}$, such that for some state-action feature $f$ and numerical threshold $c \in \mathbb{R}$:  

    \begin{equation} \label{split}
        (\pmb{x}^{(f)} < c : \forall \pmb{x} \in I_L) \cap (\pmb{x}^{(f)} \geq c : \forall \pmb{x} \in I_R)
    \end{equation}
    
    \noindent
    where $\pmb{x}^{(f)}$ is the f-th element (feature) of the state-action representation vector $x$. Let $V(P)$ be the reward sample variance of node $P$. Borrowing from XRL \cite{Liu2018TowardID} and unlike CART, we use as splitting criterion the weighted sample variance reduction $VR(P,f,c)$ of the rewards in $P$. That is:

    \begin{equation} \label{vr}
        VR(P,f,c) = V(P) - \frac{|L|}{|P|}V(L) - \frac{|R|}{|P|}V(R)
    \end{equation}    
    
    \noindent
    For each new experience sample in leaf $P$, Tree-Frontier checks for candidate partitions only in $P$. If no partition with positive $VR(P,.,.)$ exists, no split is made. Otherwise, we select the partition of $P$ that causes the highest reduction. In this sense, Tree-Frontier follows an online best-first strategy, unlike CART that follows a depth-first growth strategy.
    
    To partition the frontier samples of the splitting leaf $P$ into the two new leaves derived from $P$, we simply follow the same splitting tree rules with the experience samples. Similarly, we follow these rules to insert the new frontier samples in Tree-Frontier. To sum up, we utilize the experience samples to create splitting rules and we follow these rules to partition the frontier samples.

    Therefore, at a given time step $t$ instead of calculating the estimated Q-values of all frontier samples, we select one representative from each leaf through uniform sampling (Algorithm \ref{alg:sampling}). Then, we calculate the estimated Q-values of only those representatives and select the action with the highest Q-value. We call this procedure \textit{Tree-Frontier sampling}. We present the Tree-Frontier update algorithm for a given time step $t$ in Algorithm \ref{alg:2}. Moreover, the full TRES algorithm is presented in Appendix B. In Figure 3, we demonstrate an example of Tree-Frontier sampling. The three red nodes correspond to leaves. From each leaf, we perform uniform sampling on frontier samples (F) and then select the sample with the highest estimated Q-value ($\argmax$ operation).

    Intuitively, the use of variance reduction as the splitting criterion of Tree-Frontier intends to create leaves with small variance of rewards (which are collected from the experience samples). Since the distribution of the experience samples is the same as the distribution of the frontier samples (simply because an experience sample had been a frontier sample before it got selected by the crawler for visit), we expect that a leaf containing many experience samples and has small variance of their rewards is very probable to contain frontier samples with similar rewards that have also small variance. Therefore, although the reward of the frontier samples are not observed by the crawler unless it fetches their corresponding web pages, the RL crawler could be able to identify the leaves with frontier samples of good Q-values; i.e. the leaves from which on average it samples relevant URLs to fetch.

    \begin{theorem}[Time Complexity of Tree-Frontier sampling]\label{theorem:TreeFrontierComplexity}
        Tree-Frontier sampling has an overall time complexity of at most $P_{DDQN} (1 + \frac{3}{2}T + \frac{1}{2}T^2)$.
    \end{theorem}

\begin{proof}
The time complexity of the frontier prediction through tree-frontier update for the whole crawling process is equal to $P_{DDQN} \cdot F_{tree}$, where $F_{tree} = \sum\limits_{t=0}^{T}F_{tree}^{t}$ and $F_{tree}^{t}$ represents the number of leaves at time step $t$. Also, let $F_{tree}^0$ represent the number of leaves at time step $0$; i.e. the leaves because of splits of seed experience samples of unity rewards.  Thus, $F_{tree}^0 = 1$.  Let $d_{leaves}(t) = F_{tree}^{t} - F_{tree}^{t-1}$ be the increase of the number of tree leaves from time step $t-1$ to time step $t$. Observing that at a time step $t$ only the tree leaf that contains the experience sample of $u^{(t)}$ (the URL visited at $t$) can be split (if a significant difference based on sample variance reduction of its rewards is noticed), $d_{leaves}(t) = \mathbbm{1}{\{\textit{a split happened at t}\}} \leq 1$. Then, we can upper-bound $F_{tree}^{t}$:

            \begin{align}
            F_{tree}^{t} & = F_{tree}^0 + \sum\limits_{\tau=1}^{t}{d_{leaves}(\tau)} \leq 1 + t \nonumber
            \end{align}

\noindent
This way, we can upper-bound $F_{tree}$:

\begin{align}
    F_{tree} = \sum\limits_{t=0}^{T}F_{tree}^{t} \leq 1 + \sum\limits_{t=1}^{T} (1 + t) = 1 + \frac{3}{2}T + \frac{1}{2}T^2  \nonumber
\end{align}

\end{proof}

    \noindent
    Theorem \ref{theorem:TreeFrontierComplexity} implies that Tree-Frontier sampling has a better time complexity than the synchronous method. The fact that Tree-Frontier sampling achieves to drop the factors of $F_0$ and $\epsilon_d$ in the overall time complexity makes the RL crawling process very efficient in practice (as we also validate in our experiments next), as opposed to the computational issue of the synchronous method. 
    
    % We present the mathematical analysis of TRES in Appendix B. The full focused crawling procedure of the proposed TRES framework is presented in Algorithm \ref{alg:3} (see Appendix).

\subsection{Further Mathematical Analysis}

\subsubsection{{Preliminaries of the Mathematical Analysis}}

Let $\mathcal{X}$ be the state-action space, i.e. the space corresponding to the URL representa-tion given its path. Let $x_t^* \in \mathcal{X}$ be the state-action vector corresponding to the URL with the best Q-value in tree-frontier, according to a policy approximator function (e.g. a neural network with parameters ${\pmb{\theta}}$) at time step t. Let $n_t$ be the number of tree-frontier leaves at time step t. Let $p_{t}(x_t^*) = \frac{1}{|L_t^*|}$ be the probability that at time step t the best $x_t^*$ (according to $\hat{Q}_{\pmb{\theta}}^{(t)}$) is selected, where $L_t^*$ is the leaf containing $x_t^*$. We denote by $h(t)$ the height of the leaf containing $x^*_t$ at time step $t$. For brevity, we write the state-action value function of the agent's policy at time step $t$, produced by the function approximator, $\hat{Q}_{\pmb{\theta}}^{(t)}$ = $\hat{Q}_t$. Also, let $d_{l_i}^t$ be the distribution of i-th leaf over frontier samples at time step t.

In each node of the decision tree, let $R$ be the random variable of the reward and $\boldsymbol{X}$ the random vector corresponding to the representation of the node elements. The splitting rule is a Bernoulli variable of the form $S=\mathbbm{1}\{\boldsymbol{X}[i]>b\}$, for some feature $i$ and some threshold $b$. In particular, the chosen rule is that which maximises the empirical variance reduction calculated on the samples. Variance is calculated with respect to the reward variable $R$. Let $V^t_R(Parent)$ be the reward variance of a parent node on time step t. The parent node is split into two child nodes. Let $V^t_R(Child)$ be the reward variance of some child node on time step $t$, such that $V^t_R(Parent)=Var(R)$ and $V^t_R(Child)=Var(R|S)$.

We define $J = \sum \limits_{t=1}^{T} \mathop{\mathbb{E}}_{x_t \sim d_t}\left[\hat{Q}_t(x_t)\right]$ as a utility function, where $d_t = \frac{1}{n_t}\sum\limits_{i=1}^{n_t}d_{l_i}^t$, i.e. the state-action distribution at time step t. We note that $\hat{Q}_t$ corresponds to different policies for different values of $t$, considering a different version of the function approxi- mator for each time step. At time step $t$, we define $V_R^t(i)$ to be the variance of rewards of the experience samples of the i-th leaf and $V_q^t(i)$ to be the variance of Q-values of the frontier samples of the i-th leaf.

\subsubsection{{Objective}}
The utility function $J$ describes the total expected reward-to-go of utilizing $\hat{Q}_t$ for all time steps $t \in [0,T]$. We note that a high value of J implies selections of higher Q-values per time step, and thus policies closer to greedy ones. 
Considering that the policy approximator function $\hat{Q}_t$ is an \textit{oracle} at time step $t$, the greedy action selection is to choose the state-action pair from the tree-frontier which estimated Q-value at time step $t$ is the highest among the others available, that is $x^*_t$. Thus, oracle would always select $x^*_t$ from the tree-frontier at each time step $t$. We denote $J^* = \sum \limits_{t=1}^{T}\left[\hat{Q}_t(x^*_t)\right]$ as the oracle's utility function, i.e. the utility function in which always the best $x^*_t$ is selected at time step t.

We are interested in measuring how far from oracle  our agent is. That is, we are interested in finding a lower bound of $J - J^*$. Considering that there will always be a chance to select the best $x^*_t$ at each time step $t$, it holds that $J \leq J^*$, i.e. $J^*$ is a tight upper bound of $J$. We start by unrolling the closed form of utility function $J$. 

\begin{align}
    J &= \sum\limits_{t=1}^{T}\left[p_{t}(x_t^*)\hat{Q}_t(x^*_t) + (1-p_{t}(x_t^*)) \max_{i \in n_t}\left\{\mathop{\mathbb{E}_{x_t \sim d^t_{l_i}}} [\hat{Q}_t(x_t) \mathbbm{1}\{x^*_t \neq x_t\}] \right\} \right] \nonumber \\
    &= \sum\limits_{t=1}^{T}\left[\frac{1}{|L_t^*|}\hat{Q}_t(x^*_t) + \left(1 - \frac{1}{|L_t^*|}\hat{q}_t\right) \right] \nonumber \\
    &= J^* - \sum\limits_{t=1}^{T}\left[\left(1 - \frac{1}{|L_t^*|}\right) (\hat{Q}_t(x^*_t) - \hat{q}_t) \right] 
\end{align}

\noindent
In the equations above, we have defined 

$$\hat{q}_t = \max_{i \in n_t}\left\{\mathop{\mathbb{E}_{x_t \sim d^t_{l_i}}} [\hat{Q}_t(x_t) \mathbbm{1}\{x^*_t \neq x_t\}] \right\}$$ 

\noindent
Also, recall that the size of $L_t^*$ is bounded by the maximum number of frontier samples at time step $t$, and thus $L_t^* \leq 1 + Dt$, where $D$ is the maximum number of new frontier samples (outlinks minus deleted frontier samples) at a single time step. 

\subsubsection{Analysis}

Our goal is to find a lower bound for $J - J^*$. To achieve this we need to make some reasonable assumptions, as described below. The assumptions we make focus on some key steps of the learning algorithm and supposing these steps are successful with high probability we show that the proposed algorithm approaches the optimal performance of the oracle. These assumptions may be admittedly strong (but reasonable for focused crawling), so the analysis we make is more about explaining the logic of the algorithm and interpreting the good experimental performance.

The first assumption is about the effectiveness of variance reduction. We assume that the splitting rule uses a feature that is adequately correlated to the reward of the elements in the node being split. This assumption is practically about the quality of our data. If the web page features are weakly correlated to the reward, we can not expect any learning algorithm to perform well.   

\begin{assumption}\label{as:1}
  The  squared Pearson product-moment correlation coefficient of reward R and splitting rule S is greater than $1-\frac{\epsilon}{T\cdot h_{max}}$, where $h_{max}$ is the height limit in the tree and $T$ the number of time steps. 
\end{assumption}

\noindent
In our case, we expect that most state-action features present strong correlation to the reward. For example, we expect that actions with high relevance predictions, or with keywords in their anchor texts, are likely to be relevant to the target topic.

\begin{Proposition}
     If Assumption 1 holds, then with probability at least $1-\epsilon$ it holds that $\frac{V^t_R(Child_i)}{V^t_R(Parent_i)}<1$ for all pairs of nodes $V^t_R(Child_i),V^t_R(Parent_i)$ that lie along the path from the root to leaf $L_t^*$ for all t in $[T]$.
\end{Proposition}

The second assumption, which we denote by the \textit{Split Q-Reward Assumption}, is about the relation between the reward variance ratio and the Q-value variance ratio. Rewards are supposed to reflect real Q-values so we assume that the noise between the reward variance ratio and the Q-value variance ratio has small expectation. 

\begin{assumption}[Split Q-Reward Assumption] \label{as:2}
  In each leaf split $s$, let $Child_i$ be a new leaf created because of $s$ and $Parent_i$ be the split node because of $s$. We consider the random noise between the reward variance ratio and the Q-value variance ratio $e_i=\frac{V^t_q(Child_i)}{V^t_q(Parent_i)} -  \frac{V^t_R(Child_i)}{V^t_R(Parent_i)} $. Then, $\mathbb{E}[e_i]<\lambda\cdot exp\left(-\frac{w\sqrt{\log(T/\epsilon)}}{\sqrt{2h(t)}}\right)-1$, for each $i \in [1,h(t)]$, where $w$ is the range of $\log(1+ e_i)$ and $\lambda$ is some constant less than one. Moreover, we assume that $e_i$ are i.i.d. random variables.
\end{assumption}

\noindent
The above assumption can be well-handled in practice by manipulating appropriately the value of the discount factor $\gamma$. By setting $\gamma$ to have a relatively low value, we bring the Q-values closer to the immediate rewards. In the focused crawling setting, such assignment of $\gamma$ is quite reasonable, since we expect the agent to be myopic towards in being more concerned with maximizing immediate rewards (thus explicitly optimizing the harvest rate as pointed out in \cite{axiotis2021personalized}) and with only a few future rewards, at the time steps of which the agent could possibly identify a future hub or a better path. For example, by setting $\gamma$ equal to $0.3$, the agent is very concerned with the immediate reward, and approximately only little concerned with the next two future rewards for identifying better paths.

\begin{Proposition}
 Given Assumptions 1 and 2 it holds that $V_q(L_t^*)<\lambda^{h(t)} {V_q}^t(Root)$ for all $t$ in $[1,2,...,T]$ with probability at least $1-2\epsilon$, where $h(t)$ is the height of leaf $L_t^*$ containing $x^*_t$.
\end{Proposition}

Our third assumption is that the leaf containing $x^*_t$ has a splitting frequency greater or equal to some value $1/M_s$. In this way the height gradually increases and the Q-value variance is reduced. The first assumption guarantees that in each split effective variance reduction is achieved. The third assumption implies that this reduction takes place multiple times so that the overall reduction is high enough. The second assumption translates reward variance reduction to Q-value variance reduction. This implies that the Q-value variance in the leaf containing $x^*_t$ is not very high, so random sampling from this leaf could not select an element with a Q-value far from the optimal. Thus, the cumulative utility is bounded with high probability.

\edef\oldassumption{\the\numexpr\value{assumption}+1}
\setcounter{assumption}{0}
\renewcommand{\theassumption}{\oldassumption.\alph{assumption}}
\begin{assumption}\label{as:3a}
    The nodes of the path from the root node to the optimal leaf, $L_t^*$, were former selected leaves for at least a fraction $1/M_s$ of the times until time step $t$. Moreover, each selection leads to a split.
\end{assumption}

\noindent This is quite a strong assumption. It demands that the optimal leaf at time step $t$ has height proportional to $t$. As the time grows, this allows the existence of a constant number of long branches that repeatedly give optimal leaves. Leaves that lie at small heights are expected to contain samples with small $R$ and $Q$ values that are not selected. Such leaves end up having many samples, as splits can only occur when a leaf is selected and splits are the only way to reduce samples in a leaf. This is a key property of TRES, as samples with small $R$ and $Q$ values are accumulated in few leaves and have a small probability of being selected. The second part of the assumption, i.e. the fact that each selection leads to a split, is a mild assumption that is made mainly for technical reasons. If a selection does not lead to a split, then variance inside the leaf is probably already very small, which is an advantageous scenario. So, in the context of variance reduction assuming that each selection leads to a split does not change the worst case analysis.

\edef\oldtheorem{\the\numexpr\value{theorem}+1}
\renewcommand{\thetheorem}{\oldtheorem.\alph{theorem}}
\setcounter{theorem}{0}
\begin{theorem}
      If Assumptions 1, 2 and 3.a hold, then: 
      $$\mathbb{P}\left\{J-J^* \geq  - \frac{\sqrt{T(D+1)}}{\sqrt{\epsilon\lambda}}\cdot \frac{\Bar{\lambda}(\Bar{\lambda} + 1)( \Bar{\lambda}^T-1)}{(\Bar{\lambda} - 1)^3}
\right\} \geq 1-3\epsilon$$ 
\noindent 
where $T$ is the total number of time steps, $D$ is an upper bound for the number of new URLs inserted in frontier at each time step, parameter $\lambda$ is introduced in Assumption 2 and $\Bar{\lambda}=\sqrt[2M_s]{\lambda}$.
\end{theorem}

In the above result the absolute error is proportional to the square root of time horizon $T$. It is also proportional to the square root of $D$, where $D$ is the maximum number of new frontier samples at each time step. Finally, there is a dependence on  $\frac{1}{\sqrt[2M_s]{\lambda}-1}$, where $\lambda$ is the decrease rate of the variance of Q-values along a path from the root to a leaf.

Now we can replace Assumption \ref{as:3a}, which is quite strict, with the following milder assumption. Assumption \ref{as:3b} allows the tree to have a balanced structure and not be dominated by certain branches.

\setcounter{assumption}{1}
\renewcommand{\theassumption}{\oldassumption.\alph{assumption}}
\begin{assumption}\label{as:3b}
    The height $h(t)$ of the optimal leaf $L_t^*$ is at least $\log_2(t)$.
\end{assumption}

\setcounter{theorem}{1}
\begin{theorem}\label{t3b}
      If Assumptions 1, 2 and \ref{as:3b} hold and parameter $\lambda$ introduced in Assumption \ref{as:2} is less than $\frac{1}{8}$, then: 
      $$\mathbb{P}\left\{J-J^* \geq  - \frac{T\sqrt{T(D+1)}}{\sqrt{\epsilon}} \right\} \geq 1-3\epsilon$$ 
      \noindent
      where $T$ is the total number of time steps, $D$ is an upper bound for the number of new URLs inserted in frontier at each time step and $\Bar{\lambda}=\sqrt[2M_s]{\lambda}$.
\end{theorem}

\setcounter{theorem}{2}
\begin{theorem}\label{t3c}
     If Assumptions 1, 2 and \ref{as:3b} hold and parameter 
     $\lambda$ introduced in Assumption \ref{as:2} is less than $\frac{1}{32}$, then: 
     $$\mathbb{P}\left\{J-J^* \geq  - \frac{\sqrt{T(D+1)}}{\sqrt{\epsilon}}\cdot(\ln T + 1) \right\} \geq 1-3\epsilon$$ 
     \noindent
     where $T$ is the total number of time steps, $D$ is an upper bound for the number of new URLs inserted in frontier at each time step and $\Bar{\lambda}=\sqrt[2M_s]{\lambda}$.
\end{theorem}

Theorems \ref{t3b} and \ref{t3c} make the mild Assumption \ref{as:3b} about the leaf heights. 
% but they make strong assumptions about the decrease rate $\lambda$ of the variance of Q-values along a path from the root to a leaf. 
Theorem \ref{t3c} guarantees a sublinear, no regret kind of dependence of the error on $T$ while Theorem \ref{t3b} an error $O(T\sqrt{T\cdot D})$. All proofs of our mathematical analysis are presented in Appendix A.

\section{Experimental Setup} 
    In this section, we evaluate the effectiveness of TRES through experiments on \textit{online, real-world data} from the Web and carry out comparisons with state-of-the-art focused crawling methods. All crawling experiments are conducted \textit{in an online manner} and not on curated datasets. In particular, the experiments are on the following topic settings: (a) Sports, (b) Food, (c) Hardware. We note that these topics also belong to Open Directory Project (Dmoz)\footnote{https://dmoz-odp.org/}, which has been widely used as a reference for crawling evaluation \cite{10.1007/s11280-015-0349-x, 10.1007/s11280-015-0331-7, Elaraby2019ANA, BATSAKIS20091001}. Dmoz indexes about five million URLs covering a wide range of topics \cite{10.1007/s11280-015-0349-x}. Each of these topics is not equally represented by URLs in Dmoz; we assume a topic is as difficult to be crawled as it is to be found in Dmoz. Therefore, we expect the Hardware topic to be the most difficult of the three. 
    
    Unlike common (but less realistic) experimental settings, such as \cite{10.1145/1242572.1242632, Suebchua2017EfficientTF, bifulco2021intelligent}, that use a lot of seeds, for each experiment we utilize \textit{a single seed} and average results from 10 different single-seed crawling runs. Seeds are selected using a Google search with the condition that they are not connected to each other. In other words, none of the seeds has an outlink to any other seed.

    \subsection{Evaluation Metrics}
    For the evaluation of the crawler, we first rely on the widely used \textit{harvest rate} (HR) \cite{CHAKRABARTI19991623} of the retrieved web pages, which is denoted by the percentage of relevant retrieved web pages to the total number of all retrieved web pages. 
    More formally:
    \begin{align}
        \text{HR} = \frac{\text{Number of Relevant Retrieved Web Pages}}{\text{Number of All Retrieved Web Pages}}
    \end{align}
    Observe that under our binary rewarding schema harvest rate equals the agent's average cumulative reward. Therefore, it is equivalent to a performance metric for RL \cite{axiotis2021personalized}: measuring the harvest rate of the focused crawler is equal to measuring the average cumulative reward of the RL agent.

    Second, we evaluate the number of unique retrieved web sites (\textit{domains}), consideri- ng that a domain is relevant if it contains at least one relevant web page. That is:
    \begin{align}
        \text{Domains} = {\text{Number of Relevant Retrieved Web Sites}}
    \end{align}
    We are interested in maximizing both of the above objectives, assuming that they are equally important to the user. Moreover, we examine the computational efficiency of Tree-Frontier, in comparison to the use of both synchronous and asynchronous method \cite{10.1007/978-3-319-91662-0_20} on selecting the best action from the frontier. Last but not least, we perform extensive ablation study of the proposed framework. 

    \subsection{State-Of-The-Art and Baselines}
    We compare TRES to the following \textit{state-of-the-art} methods:
    \begin{itemize}
        \item \underline{\textit{ACHE}} \cite{10.1145/1242572.1242632}: is one of the most well-known focused crawlers, which aims to maximize the harvest rate of fetched web pages through an online-learning policy. ACHE has been widely used in crawling evaluation, such as in \cite{10.1093/bioinformatics/btt571, 10.1145/3308558.3313709}, and in crawling applications \cite{10.1145/3159652.3159724}.
    
        \item \underline{\textit{SeedFinder}} (SF) \cite{10.1007/s11280-015-0331-7}: extracts new seeds and keywords from relevant pages to derive search queries. 

       \item \underline{\textit{Asynchronous Method}} (Async) \cite{10.1007/978-3-319-91662-0_20}: an RL method that approximates the brute force by calculating the Q-values only for outlinks of the current state; the Q-values of all other links in the frontier are left unchanged. We note that this method uses only path statistical features and a relevance estimation (based on cosine similarity) within the state-action representation and no keyword informa- tion.
    \end{itemize}
    
    Furthermore, the following \textit{baseline} methods are included in our experiments: 

    \begin{itemize}
        \item \underline{\textit{TRES with Asynchronous Method}} (TRES\_Async): we evaluate the performance of TRES with the asynchronous method instead of using the Tree-Frontier sampl- ing algorithm. We are interested in measuring whether the proposed sampling algorithm increases performance on both objectives.
        \item \underline{\textit{Tree-Random}} (TR): we propose a simplified version of TRES; it selects a random action from a random leaf. In other words, Tree-Random only uses exploration and selects an action with uniform sampling from the leaf representantives (line 5 in Algorithm \ref{alg:2}). 
        % Note that the Tree-Random's frontier is the same as in TRES. We are interested in measuring the performance of Tree-Random, in order to measure the effectiveness of Tree-Frontier on discretizing the state and action spaces.
        
        \item \underline{\textit{Random}} (R): selects URLs from the frontier at random.
    
    \end{itemize}

    \subsection{Learning the reward function}
    To train the BiLSTM classifier, which will constitute the reward function, we use a small offline training dataset of web pages relevant to the target topic $C$, denoted by $D_{tr}^{(C)}$, and a large training dataset of irrelevant web pages denoted by $D_{tr}^{(C')}$ (we denote by $C'$ all topics other than target topic $C$). $D_{tr}^{(C')}$ consists of URLs from 10 Dmoz supertopics: Arts, Business, Computers, Health, News, Recreation, Reference, Science and Sports. In cases where the relevant topic, $C$, is one of the aforementioned supertopics (e.g. Sports), we remove its web pages from the training set of irrelevant web pages. One the other hand, similar to \cite{10.1145/3308558.3313709}, if the relevant topic, $C$, is a subcategory of one of these supertopics (e.g. Food and Hardware), there is only a small chance that the web pages of the supertopics belong to the target topic. Specifically, we use $1000$ and $1800$ samples for the training set of relevant web pages and for each supertopic, respectively. We perform stratified 5-fold cross validation on data set $D_{tr} = D_{tr}^{(C)} \cup D_{tr}^{(C')}$. We evaluate the performance of BiLSTM and ACHE's SVM in terms of Precision and Recall of the relevant class and Macro-average F-Score (F-Macro). As shown in Table {\ref{tab:table4}}, in all three settings, BiLSTM outperforms ACHE's SVM in all evaluation metrics. 
    % We note that BiLSTM performs worse in the Hardware setting, than the other settings, because the irrelevant class includes its supertopic, Computers, the web pages of which consistently include relevant keywords (6.67 on average).

    \subsection{Keyword Expansion Evaluation}
    We evaluate the TRES strategy to identify relevant keywords from the dataset of relevant documents, $D_{tr}^{(C)}$. Recall that TRES receives an initial small keyword set $KS$ as input, and aims to expand it with as many other relevant keywords as possible. 
    
    In all evaluated settings, we initialise the starting keyword set $KS$ with the keywor-ds provided by the Dmoz topic subcategories. We initialize $KS$ with 62, 10 and 11 keywords for each of the Sports, Food and Hardware settings, respectively. The results of keyword expansion are shown in Table {\ref{tab:table1}}. We observe that the discovered keywords (keyword set $K'$) were orders of magnitude more frequent (on average) in web pages of the target topic $C$ than in irrelevant web pages of the other topics $C'$ (web pages belonging to $D_{tr}^{(C')}$). We note that in the Food and Hardware settings $C'$ includes their supertopics (Recreation and Computers) with respective mean keyword counts equal to $0.89$ and $6.67$. Thus, the discovered keywords are orders of magnitude more frequent in $C$ than in their supertopics (if existent). 
    
    Moreover, notice that the proposed keyword expansion method discovered many new relevant keywords, especially in the Food and Hardware topics, despite the relative- ly small number of initial given keywords. On the other hand, in the Sports setting our method retrieved 32 new keywords, and thus enlarged $KS$ only by 50\%. We attribute this to the fact that in the Sports setting the number of the input keywords was six times larger than in the other evaluated settings, and as a result a candidate keyword was more difficult to produce a mean cosine similarity score higher than the defined threshold. This also supports the highly selective behavior of our keyword set expansion method, when the number of the input keywords (selected by the user) is not too small.  
    
    % \begin{table}[h!]
    % \setlength\tabcolsep{6.0pt} % default value: 6pt
    % \begin{tabularx}{\columnwidth}{@{} Z *{6}{c} @{}}
    % \toprule 
    % {Domain} & {Precision (\%)} & {Recall (\%)} & {F-macro (\%)} \\
    % \midrule 
    % Sports        & 91.3 & 90.7 & 95.4 \\ 
    % Food          & {87.0} & {90.8} & 94.2  \\
    % Hardware     & 80.0 & {86.2} & {90.8} \\
    % \bottomrule 
    % \end{tabularx}
    % \caption{Classication Results of KwBiLSTM}\label{tab:table2}
    % \end{table}

    \begin{table*}[t]
    \setlength\tabcolsep{6pt} % default value: 6pt
    \begin{tabular*}{\textwidth}{cccccccccc}
    \toprule 
    Classifier & \multicolumn{3}{c}{Sports} & \multicolumn{3}{c}{Food} & \multicolumn{3}{c}{Hardware}\\
    \cmidrule(lr){2-4} \cmidrule(l){5-7} \cmidrule(l){8-10} 
    & P & R & F-m & P & R & F-m & P & R & F-m \\
    \midrule 
    ACHE's SVM       & 76.9 & 88.5 & 90.4 & 68.5 & 81.8 & 86.7 & 62.6 & 77.0 & 83.9  \\
    BiLSTM        & \textbf{91.3} & \textbf{90.7} & \textbf{95.4} & \textbf{87.0} & \textbf{90.8} & \textbf{94.2} & {\textbf{80.0}}  & \textbf{86.2} & \textbf{90.8} \\
    \bottomrule 
    \end{tabular*}
    \caption{Classication Results: Precision (P) (\%) (relevant class), Recall (R) (\%) (relevant class) and F-Macro (F-m) (i.e. the arithmetic mean over the F$_1$-scores of the two classes) (\%)}\label{tab:table4} 
    \end{table*}

    \begin{table} [t]
        \setlength\tabcolsep{30pt}
        \begin{tabularx} {\columnwidth} {@{} Z *{3}{c} } \toprule
        {} & Sports & Food & Hardware \\ \midrule
        {$C$}  & \textbf{19.55} & \textbf{15.46} & \textbf{17.66}  \\
        {$C'$}  & 0.86  & 1.04 & 1.44   \\ \hline
        {\# $KS$}  & 62  & 10 & 11   \\
        {\# $K'$}  & 32  & 205 & 129   \\
        \bottomrule
        \end{tabularx}
        \caption{Mean Keyword Count of Discovered Keywords, Initial Keywords ($KS$), Discovered Keywords ($K'$); $C$ is the relevant topic, $C'$ are the irrelevant topics.} \label{tab:table1}
    \end{table}

    \subsection{Focused Crawling Evaluation}

    \begin{table*}[t]
        % \centering
        \setlength\tabcolsep{6pt} % default value: 6pt
        \begin{tabular*}{\textwidth}{ccccccc}
        \toprule 
        Focused Crawler & \multicolumn{2}{c}{Sports} & \multicolumn{2}{c}{Food} & \multicolumn{2}{c}{Hardware}\\
        \cmidrule(lr){2-3} \cmidrule(l){4-5} \cmidrule(l){6-7} 
        & HR & Domains & HR & Domains & HR & Domains\\
        \midrule 
        Random           & 3.55 & 17 & 8.97 &59 & 1.66 & 28  \\
        Tree-Random          & 60.01 & 204 & 36.60 & 523 & 40.63 & 354   \\ \midrule
        ACHE \cite{10.1145/1242572.1242632}        & 54.43 & {2081} & 60.45 & 2153 & 49.25 & 940  \\ 
        ACHE\_100   & 47.57 & {4574} & 50.85 & 3820 & 42.18 & 1923  \\
        ACHE\_10    & 36.79 & {{6084}} & - & - & - & -   \\  \midrule
        SeedFinder \cite{10.1007/s11280-015-0331-7}          & 56.05& 1842 & 60.16 & {3642} & 38.68 & {1314}  \\
        SeedFinder\_100     & 49.97 & 4032 & 57.49 & {5043} & 31.14 & {2255}   \\ 
        SeedFinder\_10      & 37.70 & 5774 & 46.96 & {6742} & 19.63 & 2832  \\ \midrule
        Async \cite{10.1007/978-3-319-91662-0_20}          & 47.23 & 27 & 53.12 & {36} & 30.64 & {12}  \\
        Async\_100     & 40.88 & 84 & 45.83 & {98} & 18.23 & {52}   \\ 
        Async\_10      & 25.72 & 573  & 33.19 & {654} & 5.42 & {112}  \\
        Async\_5      & 8.98 & 315  & 13.53 & {581} & ($\approx 0$) & ($\approx 0$)  \\
        Async\_5\_2    & 2.77 & 214  & 6.77 & {508} & ($\approx 0$) & ($\approx 0$)  \\ \midrule
        \textit{TRES (ours)}        & \textbf{{97.43}} & 83 & {{95.59}} & 78 & {\textbf{93.64}} & 55\\ 
        \textit{TRES\_100 (ours)}   & {94.19} & 286 & \textbf{97.55} & 301 & {88.40} & 420  \\ 
        \textit{TRES\_10 (ours)}    & {84.52} & 2335 & {94.56} & 2604 & {63.14} & 1977  \\
        \textit{TRES\_5 (ours)}     & 78.56 & 4053 & 90.98 &4324 & 46.57 & {{2956}}  \\
        \textit{TRES\_5\_2 (ours)}  & 72.43 & \textbf{7011} & 87.65 & \textbf{7304} & 42.87 & \textbf{3243}  \\
        \bottomrule 
        \end{tabular*}
        \caption{Focused Crawling Evaluation Results}\label{tab:table3}
    \end{table*}     
    
    For the focused crawling evaluation, the number of pages to be retrieved in the literature is \textit{not} fixed, but ranges significantly: from 3,000 to 10,000 pages such as in \cite{10.5555/1567281.1567456, chen2008hawk, 10.1007/978-3-319-91662-0_20}, from 10,000 to 50,000 such as in \cite{guan2008guide, 10.1145/3308558.3313709}, and above 1 million pages such as in \cite{pham2018learning}. For all experiments, we set the total number of retrieved web pages equal to 20,000. We use a BiLSTM (trained on DMOZ similar to \cite{10.1145/3308558.3313709}) as ground-truth to classify the retrieved web pages of each method. We note that all crawling experiments are carried out in an online manner using real-world data from the Web. In other words, for each fetched web page, we add to the frontier all its real outlink URLs from the Web and not only those belonging to a curated dataset (such as the dataset provided by Dmoz). 

    \begin{figure}[t]
        \centering
        \includegraphics[scale=0.32]{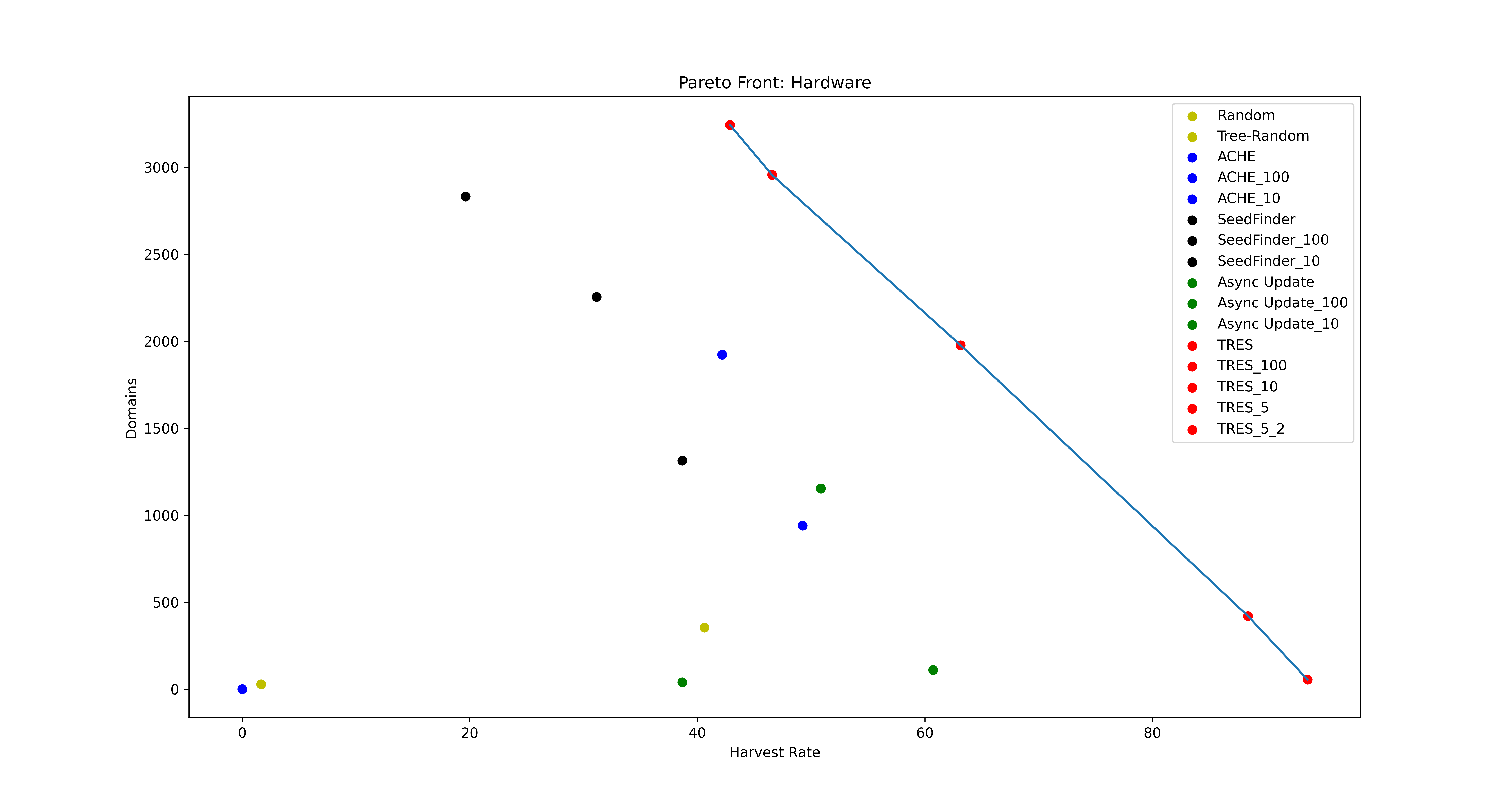}
        \caption{The Pareto front of the examined methods in the Hardware setting. An ideal focused crawler would correspond to a point at top right, that is a point maximizing both objectives.}
        \label{pareto_hardware}
    \end{figure}
    
    Table {\ref{tab:table3}} presents the focused crawling results for the three topic settings, in terms of harvest rate and the number of retrieved domains. To force crawlers to be able to maximise either objective, we consider the constraint that the crawler can retrieve at most $MAX$ web pages from a domain (web site), and we divide our experiments into the following five categories: (1) $MAX = \infty$ (i.e. 20,000), (2) $MAX = 100$, (3) $MAX = 10$, $MAX = 5$, and adaptive $MAX$ starting from $5$ and reaching $2$ (which we denote by 5\_2). 
    
    As shown in Table {\ref{tab:table3}}, for all evaluated values of $MAX$, our TRES consistently outperforms by a large margin the other methods, in terms of harvest rate. Moreover, by lowering $MAX$, TRES manages to constantly increase the number of retrieved domains (while maintaining a good harvest rate), because by doing so the crawler is forced to explore new web sites instead of fetching relevant web pages from already retrieved relevant domains. We observe that TRES\_5\_2 significantly outperforms any other evaluated method in terms of the number of retrieved domains on all evaluated settings. Therefore, TRES manages to Pareto-dominate all other evaluated methods in terms of both objectives in all three test settings. In Figure \ref{pareto_hardware}, we demonstrate the Pareto front in terms of the two objectives in the Hardware setting. As we can clearly see from the figure, TRES significantly improves the Pareto front of the examined methods. We note that a different assignment of the $MAX$ value corresponds to a different point of the Pareto front; as highlighted above, lower values of $MAX$ lead to more domains and lower harvest rate, while higher values of $MAX$ lead to the opposite results.   

    We note that in $MAX \leq 10$, ACHE aborted, as it run out of URLs to crawl, and thus we could not provide its average results, since it did not retrieve the defined number of web pages. Similar experimental problems of ACHE have been observed in previous works \cite{10.1145/3308558.3313709}, e.g. when the method is bootstrapped by a few seeds. 
    
    % Specifically, it outperforms the state-of-the-art by at least: 2\% in $MAX = \infty$, 10\% in $MAX = 100$, 13\% in $MAX = 10$, 13\% in $MAX = 10$, and 13\% in $MAX = 5\_2$. 
    
    Also, our baseline, Tree-Random, performs on a par with ACHE, SeedFinder and Async in harvest rate when $MAX = \infty$. In particular, Tree-Random outperforms the state-of-the-art methods in the Sports setting and SeedFinder and Async in the Hardware setting.  The lower harvest rate scores of the evaluated methods could be attributed to the fact that these methods quickly level off when fewer relevant web sites are found, since they rely on the relevant web sites that are already discovered \cite{10.1145/3308558.3313709}. Moreover, Async's low performance is due to the weak state representation that it adopts, as we highlight next in the ablation study.   

    Comparing the harvest rates of Tree-Random and the Random crawler in the Sports topic, we observe that the former achieves $60.01\%$ and the latter $3.55\%$. This means that the true percentage of relevant URLs were on average $3.55\%$ in the frontier, yet the probability of selecting a leaf of Tree-Frontier leading to a relevant URL is $60.01\%$. Thus, in such a scenario, Tree-Frontier generated a lot of leaves, in which the agent demonstrated on average the desirable behavior, despite the fact that there were on average 27 times more actions (assuming a similar frontier distribution with Random crawler), leading to irrelevant URLs, in the frontier. We also note that such a good performance of Tree-Random also validates our intuition for the effectiveness of the variance reduction as a splitting criterion for Tree-Frontier.

    \subsection{Ablation Study}

    As we discussed in the previous subsection, TRES achieves state-of-the-art results in the two crawling objectives. In this subsection, we are interested in conducting an ablation study in order to assess various aspects of TRES in detail. 
    
    \subsubsection{Tree-Frontier vs Synchronous / Asynchronous Method}
    The first question we aim to address is whether the use of Tree-Frontier indeed enhances crawling performance. As we have already discussed, the use of Tree-Frontier updates facilitates efficient sampling of the frontier, accelerating by orders of magnitude the computationally infeasible synchronous method, which exhaustively examines every frontier sample at
    each time step, while allowing a good performance. Figure \ref{fig:frontier} illustrat-es the behavior of TRES as crawling progresses and aims to explain why the synchronous method is computationally infeasible in practice. We observe that the size of the subset of the frontier that TRES examines at each time step is orders of magnitude smaller than the full frontier size. In particular, at 20,000 time steps the frontier contains approximately 800,000 URLs, all corresponding to candidate web pages for the crawler to fetch in the next crawling step. On the other hand, the Tree-Frontier has only 2,500 leaves. This means that at that crawling step (i.e. the 20,000th step), with Tree-Frontier sampling the RL crawler ought to find the best URL out of 2,500 representative URLS in the frontier, while with the synchronous method the RL crawler ought to examine all 800,000 URLs of the frontier. Proportionally, this also means that at each time step with Tree-Frontier sampling the RL crawler ought to examine approximately 320 times less URLs than with the synchronous method, which makes the latter computationally infeasible to run for long crawling sessions in practice, also taking into account the constantly increasing frontier size over time. It is worth noting that these experimental findings validate and strengthen the importance of Theorems 1 and 2 about the overall time complexity of our method and the synchronous one. 

    \begin{figure*}[t]
        \begin{minipage}{.5\linewidth}
        \centering
        \subfloat[]{\label{main:a}\includegraphics[scale=.4]{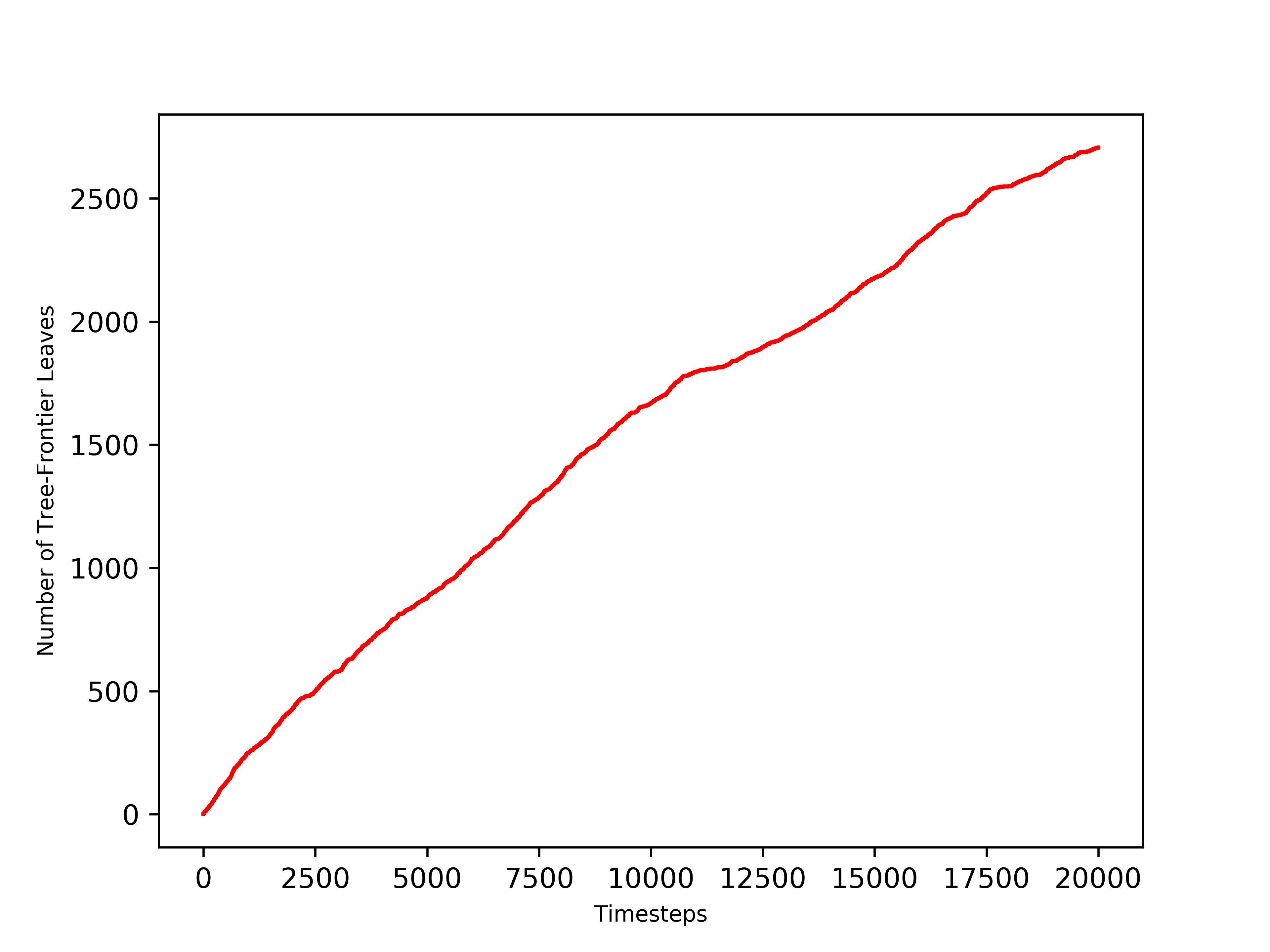}}
        \end{minipage}%
        \begin{minipage}{0.5\linewidth}
        \centering
        \subfloat[]{\label{main:b}\includegraphics[scale=.4]{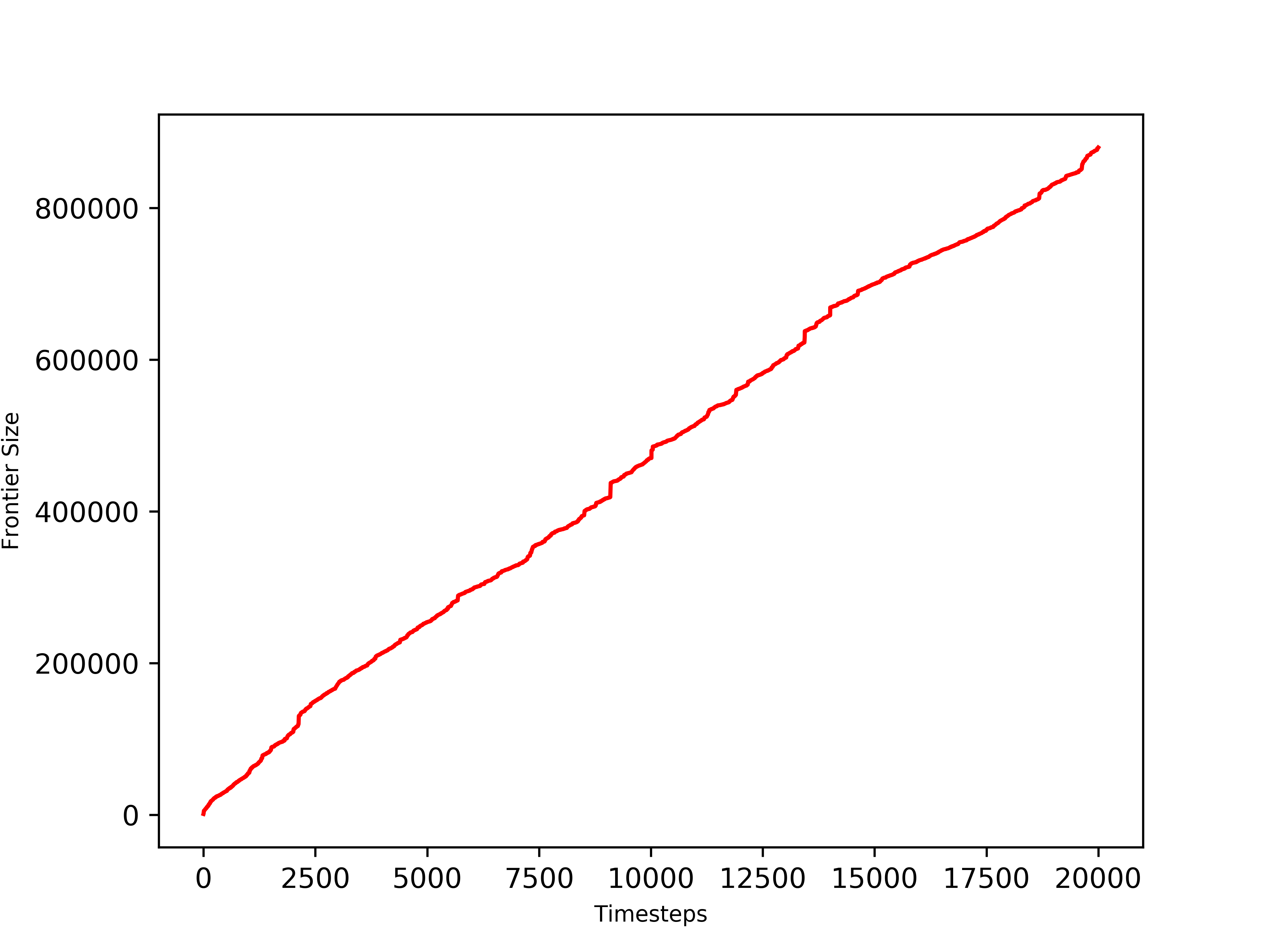}}
        \end{minipage}\par\medskip
        % \centering
        % \subfloat[]{\label{main:c}\includegraphics[scale=.5]{frontier.png}}
        \caption{Frontier Sampling: (a) Number of Tree-Frontier leaves over time and (b) Frontier size over time. The figures illustrate that with Tree-Frontier sampling (corresponding to left figure) the RL crawler ought to examine approximately 320 times less URLs than with the synchronous method (corresponding to right figure), which makes the latter computationally infeasible to run for long crawling sessions in practice}
        \label{fig:frontier}
    \end{figure*}

    % In Fig. \ref{fig:frontier}, we can also observe that the dependence of the frontier size and the number of leaves (in Tree-Frontier) to the number of time steps is not exactly linear. In the case of the full frontier (synchronous method) this is attributed to the variance in the number of outlinks that a web page could contain. There is an oscillation around the mean outlink number, which would ideally be equal to the slope of the curve (see Fig. \ref{fig:frontier}, subplot (b)), had there been no variance in the random variable of the number of outlink URLs (of a web page). In the case of Tree-Frontier (see Fig. \ref{fig:frontier}, subplot (a)), the variance of the slope is due to the fact that at some time steps no split is performed and, thus, the number of leaves remains the same. At other time steps, binary splits may occur, increasing the size of the set of the candidate frontier samples that we examine. 

    We are also interested in measuring the advantage of using the Tree-Frontier over the asynchronous method of \cite{10.1007/978-3-319-91662-0_20}. In Table \ref{tab:ablation}, we have also included experiments (with $MAX \leq 5$) of TRES which adopts the asynchronous method instead of the Tree-Frontier (referred as TRES\_Async). As can be clearly seen, the use of Tree-Frontier improves both the harvest rate and the number of retrieved domains by at least 3\% and 300 domains, respectively, on all evaluated settings. Moreover, TRES' state representation significantly improves both the harvest rate and the number of retrieved domains over Async (crawling algorithm of \cite{10.1007/978-3-319-91662-0_20}) on all evaluated settings and values of $MAX$. We attribute the ineffectiveness of Async to its proposed state representation which lacks any keyword information but also makes use of less informative relevance estimation via a simple cosine similarity; in contrast to the more accurate relevance predictions of BiLSTM that TRES/TRES\_Async exploit.

    %%%%% FUTURE WORK %%%%%%
    % It is worth mentioning that the Tree-Frontier update approach could be also employed in other RL settings. Tree-Frontier effectively tackles the exponential increase in the size of the action space, by reducing the number of actions that are examined at each time step, that is the number of samples for which Q-values are calculated. A Tree-Frontier update is performed in a stratified way through an online discretization of both large state and action spaces. Therefore, in an RL setting with rampant action space growth (frontier), the Tree-Frontier algorithm could potentially be preferable to an exhaustive calculation of Q-values.

    \subsubsection{Evaluating ACHE's SVM in TRES}
    Furthermore, we aim to address the important question of how the selected classifier (that constitutes the RL reward) affects the crawling performance. To this end, we compare TRES to TRES\_SVM which makes use of ACHE's SVM as the reward function of its RL online training on all evaluated settings with $MAX \leq 10$. To make the comparisons fair, after we have collected the resulted fetched URLs, we evaluate all crawls with the same ground-truth classifier (i.e. the BiLSTM classifier that was also used in the original evaluation). As Table \ref{tab:ablation} also demonstrates, TRES with BiLSTM performs consistently better than TRES\_SVM; however the latter performs on a par with the former on the Sports and Food settings, but significantly worse on the Hardware setting. We attribute this to the fact that ACHE's SVM has significantly lower precision (but good recall) in Hardware than in Sports and Food, and therefore TRES\_SVM retrieves many irrelevant web pages that it erroneously reckons to be relevant. We also observe that TRES\_SVM\_10 performs better in terms of harvest rate than ACHE and SeedFinder (which use SVM in their implementation) on all settings, while baseline TRES\_SVM\_5\_2 achieves slightly worse performance in the number of retrieved doma-ins.

    \begin{table*}[t]
        % \centering
        \setlength\tabcolsep{6pt} % default value: 6pt
        \begin{tabular*}{\textwidth}{ccccccc}
        \toprule 
        Focused Crawler & \multicolumn{2}{c}{Sports} & \multicolumn{2}{c}{Food} & \multicolumn{2}{c}{Hardware}\\
        \cmidrule(lr){2-3} \cmidrule(l){4-5} \cmidrule(l){6-7} 
        & HR & Domains & HR & Domains & HR & Domains\\
        \midrule 
        TRES\_SVM\_10    & {{79.88}} & 1692 & {{87.65}} & 2094 & {34.64} & {1189}  \\ 
        TRES\_10\_NH    & {82.17} & {\textbf{2393}} & 94.22 & {\textbf{3189}} & 58.89 & {\textbf{2080}}   \\ 
        \textit{TRES\_10 (ours)}    & {\textbf{84.52}} & 2335 & {\textbf{94.56}} & 2604 & {\textbf{63.14}} & 1977  \\ \midrule
        TRES\_Async\_5      & 74.17 & 3788\ & 87.26 & {3719} & 42.34 & {2550}  \\
        \textit{TRES\_5 (ours)}     & \textbf{78.56} & \textbf{4053} & \textbf{90.98} & \textbf{4324} & \textbf{46.57} & {{\textbf{2956}}}  \\ \midrule
        TRES\_SVM\_5\_2   & {{60.23}} & {5783} & {{70.39}} & {5342} & {4.4} & {318}  \\
        TRES\_Async\_5\_2      & 65.55 & 6595\ & 83.43 & {7012} & 34.34 & {2649}  \\  
        \textit{TRES\_5\_2 (ours)}  & \textbf{72.43} & \textbf{7011} & \textbf{87.65} & \textbf{7304} & \textbf{42.87} & \textbf{3243}  \\
        \bottomrule 
        \end{tabular*}
        \caption{Ablation Study of TRES}\label{tab:ablation}
    \end{table*}

    \subsubsection{Effect of the web domain features}
    The use of web domain features is based on the intuition that the focused crawler is more likely to fetch more relevant web pages by trying to avoid less relevant or unknown web sites (domains). In Table \ref{tab:ablation}, we show the trade-off between maximizing the harvest rate of web pages and maximizing the number of different domains. We evaluate TRES with $MAX=10$ with and without the web domain features (\textit{TRES\_10} and \textit{TRES\_10\_NH}, respectively). We present the evaluation on the $MAX=10$ setting, where TRES performs on a par with the other methods in the domain maximization task. We observe that the use of Hub features boosts the harvest rate performance in all evaluated topic settings, while the absence of these features helps increase the number of different domains. 
    
    Recall that web domain features describe the expected relevance of a domain (web site) for the current iteration of the algorithm. The use of web domain features leads the crawler to choose more likely relevant URLs from the frontier, ignoring the fact that they may belong to a fetched domain. On the other hand, their absence allows the crawler to be more explorative towards new domains, trying to identify new relevant locations on the Web.

\section{Related Work}

The problem of discovering interesting content on the Web has received a surge of approaches in the literature. These approaches are often categorized in two classes  \cite{10.1145/3308558.3313709}: (a) \textit{web crawling methods} and (b) \textit{search-based discovery methods}. 

Web crawling methods utilize the link structure to discover new content by automat-ically fetching and iteratively selecting URLs extracted from the visited web pages. They can be divided into: (a) \textit{backward} and (b) \textit{focused} crawling methods. Backward crawling methods, e.g. the BIPARTITE crawler \cite{barbosa-etal-2011-crawling}, explore URLs through backlink search. In this work, we are not interested in evaluating backward crawling (or methods that may use it), since this type of crawling is impractical requiring access to paid APIs for most real users.

Since the introduction of focused crawling by Chakrabarti et al. \cite{CHAKRABARTI19991623}, several improv-ements have been made to the basic process. In this paper, we are interested in improv-ements related to \textit{learning-based} approaches. Many approaches, such as \cite{6165295, RAJIV20213, SALEH2017181, Elaraby2019ANA, shrivastava2022efficient}, use  classifiers in order to determine the crawling strategy by retrieving URLs from the frontier. Suebchua et al. \cite{Suebchua2017EfficientTF} introduced the ``Neighborhood feature'' in focused crawling, which utilizes the relevance of all already fetched web pages of the same domain (web site), in order to effectively find URLs belonging to the same domain. However, a certain limitation is that their approach of identifying hubs highly depends on how many relevant URLs have been found on the already crawled web sites.
ACHE \cite{10.1145/1242572.1242632} is an adaptive crawler aiming to maximize the harvest rate \cite{CHAKRABARTI19991623} of fetched web pages through an online-learning policy. It learns to identify promising URLs and, as the crawl progresses, it adapts its focus relying on the relevant web sites that it has fetched. Nevertheless, ACHE's weak classifier and its quite simple online learning approach make it prone to weak performance and even to abort unexpectedly when only a few seeds are used as input. 

Focused crawling methods often utilize natural language processing techniques to benefit from the semantic information of keywords or related text. Du et al. \cite{10.1016/j.asoc.2015.07.026} integrated TF-IDF scores of the words to construct {topic and text semantic vectors} and compute cosine similarities of these vectors to select URLs. Deep-learning-based methods \cite{dhanith2021word, neelakandan2022automated, kuehn2023threatcrawl} use appropriate word embeddings and deep neural networks to calculate semantic similarity scores between web pages and the topics of interest. Agre et al. \cite{agre2015keyword} adopted keyword-driven focused crawling with relevancy decision mechanisms that used ontology concepts for improving crawler's performance. Liu et al. \cite{liu2022applying} proposed an ontology learning and multi-objective ant colony optimization method based on TF-IDF similarity scores for the construction of ontology topic models for focused crawling. However, all these approaches fail to explicitly parameterize the crawler's strategy, so that the crawler to explore effectively different but informative regions of the Web, whereas they are heavily based on prior knowledge and heuristic search criteria. 

Reinforcement Learning (RL) can be a natural way to parameterize the focused crawler's strategy (in the form of a learned policy) for optimizing focused web crawling. Rennie and McCallum \cite{Rennie99efficientweb} introduced RL in focused crawling, but in their RL framew-ork each action was immediately disconnected from the state. InfoSpiders \cite{Menczer99adaptiveretrieval, 10.1145/1031114.1031117} uses policy approximation based on a heuristic state representation but forces the crawler agent to follow a single path of web pages leading to poor performance. Han et al. \cite{10.1007/978-3-319-91662-0_20} proposed an MDP formulation for focused crawling, with handcrafted shared state-action representation, and solved it with SARSA \cite{10.5555/3312046}. However, their MDP formulation somewhat unnaturally allows the agent to deviate from the successive order of states; i.e. it may select actions that do not really exist in a current state. To select actions from a greedy RL crawling policy, they introduced the synchronous method; at each time step they updated the Q-values of all actions (URLs) and selected the URL with the highest value. To overcome the computational cost of the synchronous method due to the large action space, they also proposed the asynchronous method which updates the Q-values of only those actions belonging to the current state; which however makes it diﬃcult for the crawler agent to be near-optimal with respect to its greedy RL policy. Axiotis et al. \cite{axiotis2021personalized} used deep Q-learning to train an RL crawler in an online active learning manner, adopting the synchronous method for estimating the URLs of the frontier. The reason why they managed to use the synchronous method, without noticing its infeasibility issue, was the very short crawling horizon they set in their problem (i.e. they used $T<1000$), which maintained the increasing frontier size quite small.

% However, in the asynchronous method the Q-values of the URLs in the frontier are calculated at different time steps making it diﬃcult for the crawler agent to be near-optimal with respect to its greedy RL policy. 

% RLwCS \cite{10.5555/1567281.1567456} uses tabular Q-learning to select the best classifier from a defined pool of classifiers, all of which evaluate the candidate URLs of the frontier. 

Search-based discovery methods utilize queries on search engine APIs and form a strategy for query improvement. SeedFinder \cite{10.1007/s11280-015-0331-7} tries to find new seeds in order to best bootstrap the crawling process of an arbitrary focused crawler, by executing queries on search engines utilizing terms (keywords) extracted from relevant texts. 
% Liakos et al. \cite{10.1007/s11280-015-0349-x} expanded a keyword collection by selecting the keyword with the highest TF-IDF score to form a query. By executing this query on a search engine, new candidate keywords were retrieved and the ones with the highest scores were stored in the collection. 
Pham et al. \cite{10.1145/3308558.3313709} proposed DISCO which uses a multi-armed bandit strategy to combine different search operators (queries on APIs, forward and backward crawling) as proxies to collect relevant web pages, and rank them, based on their similarity to seeds. Zhang et al. \cite{zhang2021dsdd} used search engines and online learning to incrementally build a model that recognizes web sites that contain datasets with relevant context. Similar to DISCO, this method utilizes a set of discovery operators (including backward crawling) to improve the search, and applies a multi-armed bandit algorithm as well for selecting the best operator.

To face certain challenges arising in the above literature, in this paper we propose TRES, a novel end-to-end RL focused crawler which does not require access to any possibly paid API. TRES uses RL to learn an effective crawling strategy, that never aborts unexpectedly even when a single seed is used as input. The proposed framework is built on top of a novel, mathematical sound MDP formulation that allows the crawler agent to both expand multiple web paths of the Web graph and preserve the successive order of state transitions by always selecting actions that are explicitly available in the current state. To face the infeasible computational cost of the synchronous method during the action selection, we propose Tree-Frontier sampling, an algorithm that adaptively discretizes the large state-action joint space and selects the best action out of only a few representative ones. Moreover, TRES introduces informative features for representating the states, such as web domain and keyword features, that implicitly encourage the crawler to better explore interesting locations of the Web. TRES also makes use of pretrained word embeddings in order to effectively identify informative keywords for explicitly enriching the action representation of the candidate URLs and facilitate the discovery of relevant web pages, without utilizing any external search engine, or explicitly constructing any ontology concepts. 

% From a more theoretical standpoint, Gouriten et al. \cite{10.1145/2631775.2631795} highlighted the frontier batch disadvantage, considering that a batch is the number of successive crawling time steps in which the frontier has not been updated. They showed that even in an offline setting, the problem of determining the optimal sequence of fetched pages, is NP-hard. To tackle this problem, RL has been employed to find a policy in order to select which web pages to fetch from the frontier.
 
% To face certain challenges arising in the above literature, we propose TRES, a novel RL-based focused crawling framework. Similar to \cite{10.1007/978-3-319-91662-0_20}, we formulated the focused crawling setting as an MDP, but we achieved to allow the agent to both preserve the successive order of state transitions and avoid the ``crawling traps'', in which the agent only follows the URLs that are hyperlinks of the last crawled web pages. We encourage a more exploratory crawling strategy through the use of statistical features related to keywords, relevance estimation and the web path. To deal with the large URL collection that is stored in the frontier (large state-action space), we introduce the Tree-Frontier method that efficiently organizes this collection and allows subtle sampling from it. 

\section{Conclusion} 
    In this paper, we proposed an end-to-end reinforcement learning focused crawling framework that outperforms state-of-the-art methods on the harvest rate and the number of retrieved web sites (domains). We formalized focused crawling as a Markov Decision Process, in which the agent necessarily follows the successive order of state transitions and selects actions (related to URLs) that are explicitly available in the current state (related to the Web graph). In addition, we improved the impractical time complexity of the synchronous method through the use of Tree-Frontier that we introduced. A Tree-Frontier update handles the large state and action spaces through adaptive discretization, while at the same time facilitating efficient sampling of the frontier, improving by orders of magnitude the computational cost issue of the synchronous method.
    
    As a future work, we would like to investigate the addition of an intrinsic reward, aiming to enhance the agent's exploration for visiting a larger number of relevant domains. Moreover, we intend to effectively design the reward function towards giving more emphasis on exploring specific regions of the Pareto front (e.g. focusing more on the harvest rate but also fetching a specific number of domains). Such an improvement on the proposed framework would further increase its usefulness by capturing a wider range of real-time use cases.
    An additional open challenge is the minimization of the size of the labeled data set needed to train the classifier, which could be used to estimate the reward function.

 % \bibliographystyle{elsarticle-num}
 \bibliographystyle{apalike}
 \bibliography{cas-refs}

\newpage

\appendix

\section{Missing Proofs}
We first prove the two auxiliary propositions that use Assump-tions 1 and 2. Then we derive a lemma that provides a generic bound about the suboptimality of the objective. Finally, we customise this bound for different cases of the tree structure through Assumptions 3.a, 3.b and 3.c.\\ \\
Using Assumption 1 we prove Proposition 1: \\\\
\textit{Statement:} If Assumption 1 holds, then with probability at least $1-\epsilon$ it holds that  $\frac{V^t_R(Child_i)}{V^t_R(Parent_i)}<1$ for all pairs of nodes $V^t_R(Child_i),V^t_R(Parent_i)$ that lie along the path from the root to leaf $L_t^*$ for all t in $[1,2,...T]$.
\begin{proof}
In each node of the decision tree let $R$ be the random variable of the reward and $X$ the random vector corresponding to the representation of the node elements. The splitting rule is a Bernoulli variable of the form $S=\mathbbm{1}\{X[i]>b\}$, for some feature $i$ and some threshold $b$. In particular, the rule is chosen which maximises the empirical variance reduction calculated on the samples. Variance is calculated with respect to the reward variable R. Let $V^t_R(Parent)$ be the reward variance of a parent node on timestep t. The parent node is split into two child nodes. Let $V^t_R(Child)$ be the reward variance of some child node on timestep t. $V^t_R(Parent)=Var(R)$ and $V^t_R(Child)=Var(R|S)$. We will use an inequality that holds for the conditional variance of two correlated variables in order to bound the child variance from above. $\mathbb{E}[Var(y|R)]\leq Var(R)-\frac{Cov^2(R,S)}{Var(S)}$. If $Var(R)=0$ then $\mathbb{E}[Var(R|S)]$ is zero as well. Else we have
% $\mathbb{E}[\frac{Var(y|R)}{Var(y)}]\leq 1-\frac{Cov^2(y,R)}{Var(R)\cdot Var(y)} \leq 1-16Cov^2(y,R)$, because both y and R are Bernoulli random variables, so their variance is at most $max_{p \in (0,1)}\{p(1-p)\}=\frac{1}{4}$. We assume $|Cov(y,R)|\geq C$.
$\mathbb{E}[\frac{Var(R|S)}{Var(R)}]\leq 1-\frac{Cov^2(R,S)}{Var(S)\cdot Var(R)} \leq \delta$, assuming that the  squared Pearson product-moment correlation coefficient of R and S is greater than $1-\delta$. This hypothesis is reasonable, because we expect that the splitting rule which causes the greatest variance reduction is highly correlated to the reward variable.\newline
Applying Markov's inequality we obtain $\mathbb{P}[\frac{Var(R|S)}{Var(R)}\geq a] \leq \frac{\mathbb{E}[\frac{Var(R|S)}{Var(R)}]}{a}\leq \frac{\delta}{a}$. We set $a=1$.
Thus, with probability at least $1-\delta$ the ratio of variance between the child and the parent node is less than one. We want the variance ratio condition to hold for all timesteps along the path leading to the optimal leaf node. Let $h(t)$ be the height of the leaf containing $x^*_t$ at timestep $t$. We set $T\cdot h_{max}\cdot\delta=\epsilon$ and taking a union bound the condition is satisfied for all timesteps with probability at least $1-\epsilon$. This completes the proof. \end{proof}
\noindent Using the result in Proposition 1 along with Assumption 2 we can prove Proposition 2:\\\\
\textit{Statement:} 
     Given Assumptions 1 and 2 it holds that $V_q(L_t^*)<\lambda^{h(t)} {V_q}^t(Root)$ for all $t$ in $[1,2,...,T]$ with probability at least $1-2\epsilon$, where $h(t)$ is the height of leaf $L_t^*$ containing $x^*_t$.
\begin{proof}
% If the event described in Proposition 1 is true and Assumption 2 holds, then with probability at least $1-\epsilon$ it holds that $V_q(L_t^*)<\lambda^{h(t)} {V_q}^t(Root)$ for all $t$ in $[1,2,...,T]$,. Also, the event described in Proposition 1 is realised with probability at least $1-\epsilon$, given that Assumption 1 holds.
% The Split Q-Reward Assumption implies that inside a leaf, similar rewards correspond to similar Q-values (produced by $\hat{Q}_t$). We consider that all high rewards (ones if the reward function is binary) come from high Q-values. In other words, in our crawling setting, we consider that our agent has higher Q-values for the most URLs corresponding to relevant web pages. 
Considering that $Root$ is the tree root node and $h(t)$ is the height of $Child$ at timestep t, it holds that 

\setcounter{equation}{1}
\begin{align}
    V^t_q(Child) &= \left(\frac{V^t_R(Child)}{V^t_R(Parent)} + e_{h(t)}\right) \cdot V^t_q(Parent) \nonumber \\
    &= \prod_{i=1}^{h(t)}\left(\frac{V^t_R(Child_i)}{V^t_R(Parent_i)} + e_i\right) \cdot V^t_q(Root) \nonumber \\
    &= \prod_{i=1}^{h(t)}\left(\frac{Var(R_i|S_i)}{Var(R_i)}+ e_i\right) \cdot V^t_q(Root) \nonumber \\
    &< \left( \prod_{i=1}^{h(t)}\left(1+ e_i\right) \right)\cdot V^t_q(Root) \label{eu_eqn} ,\\
    &\text{\quad for all $t$ in $[1,2,...,T]$, with probability at least $1-\epsilon$} \nonumber
    % \\ 
    % &< (1 + \epsilon)^{h(t)} \cdot V^t_q(Root)
\end{align}
We consider the noise variables $l_i=\log(1+ e_i)$. Since the $e_i$ variables are i.i.d. the same holds for the $l_i$ variables. We have assumed that their range is upper bounded by $w$. Their mean can be bounded by applying Jensen's inequality: $\mathbb{E}[l_i]=\mathbb{E}[\log(e_i+1)]\leq log(\mathbb{E}[e_i]+1)<log(\lambda)-\frac{w\sqrt{\log(T/\epsilon)}}{\sqrt{2h(t)}}:=l$, since $ \mathbb{E}[e_i]<\lambda\cdot exp\left(-\frac{w\sqrt{\log(T/\epsilon)}}{\sqrt{2h(t)}}\right)-1 \quad\forall i \in [1,h(t)]$.\\
From Hoeffding's inequality we have that $$\mathbb{P}\left[\sum\limits_{i=1}^{h(t)}l_i-h(t)l\geq s\right]\leq\mathbb{P}\left[\sum\limits_{i=1}^{h(t)}l_i-h(t)\mathbb{E}[l_i] \geq s\right]\leq  exp \left( \frac{-2s^2}{h(t)w^2}\right)$$ where $w$ bounds the range of $l_i$ and $s>0$.  \\
We set $s=s_{\epsilon}:=\frac{w\sqrt{h(t)\cdot \log(T/\epsilon)}}{\sqrt{2}}$. Then we have $\mathbb{P}\left[\sum\limits_{i=1}^{h(t)}l_i-h(t)l>s_{\epsilon}\right]\leq T/\epsilon$. Taking a union bound over $t\in[1,...T]$ we obtain that with probability at least $1-\epsilon$ we have for all $t\in[1,...T]$: $$\sum\limits_{i=1}^{h(t)}l_i-h(t)l<s_{\epsilon} \Leftrightarrow \sum\limits_{i=1}^{h(t)}\log(1+ e_i)<h(t)\cdot \log(\lambda)\Leftrightarrow\prod\limits_{i=1}^{h(t)}\left(1+ e_i\right)<\lambda^{h(t)}$$
Combining the above with inequality \ref{eu_eqn} we obtain the desired result, that is $V_q(L_t^*)<\lambda^{h(t)} {V_q}^t(Root)$ for all $t$ in $[1,2,...,T]$.\\
The inequality $\prod\limits_{i=1}^{h(t)}\left(1+ e_i\right)<\lambda^{h(t)}$, $t\in[1,...T]$ holds with probability at least $1-\epsilon$ and inequality \ref{eu_eqn} holds for all $t$ in $[1,2,...,T]$, with probability at least $1-\epsilon$. Thus, from union  bound, they both hold with probability at least $1-2\epsilon$. The desired result only depends on these two inequalities and thus it holds with  probability at least $1-2\epsilon$.\\
\end{proof}
Next we will prove an auxiliary lemma that uses Assumptions 1 and 2 and provides a generic bound about the suboptimality of the objective.
\newtheorem{lemma}{Lemma}
\begin{lemma}
If Assumptions 1 and 2 hold, then $$\mathbb{P}\left\{J-J^* \geq  - \frac{\sqrt{T}}{\sqrt{\epsilon}}\cdot\sum\limits_{t=1}^{T}\left[\left(\sqrt{Dt^3+t^2} \right){\sqrt{\lambda}}^{h(t)} \right] \right\} \geq 1-3\epsilon$$ where $T$ is the total number of timesteps, $D$ is an upper bound for the number of new URLs inserted in frontier at each timestep and parameter $\lambda$ is introduced in Assumption 2.    
\end{lemma}
\begin{proof}
Let $d_{L^*_t}$ be the distribution Q-values of $L_t^*$, the leaf which contains $x^*_t$ at timestep t. From Chebyshev's inequality we obtain:\\
$\mathbb{P}[|\hat{q}_t-\mu_{L^*_t}| \geq k\sigma_{L^*_t}] \leq \frac{1}{k^2}$
Thus, from union bound, with probability at least $1-\frac{|L^*_t|}{k^2}$ all frontier samples of $L_t^*$ lie within the interval $[\mu_{L^*_t}-k\sigma_{L^*_t},\ \mu_{L^*_t}+k\sigma_{L^*_t}]$. The interval has length $2k\sigma_{L^*_t}$, so the difference between any two elements of the  interval is at most $2k\sigma_{L^*_t}$. We set $k=\sqrt{\frac{|L^*_t|}{\epsilon_0}}$, so the probability that $\hat{Q}_t(x^*_t) - \hat{q}_t \leq 2k\sigma_{L^*_t}$ is at least $1-\epsilon_0$. We also set $T\cdot\epsilon_0=\epsilon$, so that the condition holds for all timesteps with high probability (from union bound over T timesteps).\\
$\mathbb{P}\left\{\hat{Q}_t(x^*_t) - \hat{q}_t \leq 2\sqrt{\frac{|L^*_t|}{\epsilon_0}}\sigma_{L^*_t}\right\} \geq 1-\epsilon_0 \Rightarrow
\mathbb{P}\left\{\hat{Q}_t(x^*_t) - \hat{q}_t \leq 2\sqrt{\frac{T(D\cdot t+1)}{\epsilon}}\sigma_{L^*_t}\right\} \geq 1-\frac{\epsilon}{T}$\\ \\
Recall that $J=J^* - \sum\limits_{t=1}^{T}\left[\left(1 -  \frac{1}{|L_t^*|}\right) (\hat{Q}_t(x^*_t) - \hat{q}_t) \right]$. 
Using the probabilistic bound for $\hat{Q}_t(x^*_t) - \hat{q}_t$ we yield:
\begin{align}
\mathbb{P}\left\{      J-J^* \geq  - \sum\limits_{t=1}^{T}\left[2\left(1 - \frac{1}{|L_t^*|}\right) \sqrt{\frac{|L^*_t|}{\epsilon_0}}\sigma_{L^*_t} \right] \right\} \geq 1-\epsilon \Rightarrow \nonumber\\
\mathbb{P}\left\{      J-J^* \geq  - \sum\limits_{t=1}^{T}\left[2\left(1 - \frac{1}{Dt+1}\right) \sqrt{\frac{T(Dt+1)}{\epsilon}}\sigma_{L^*_t} \right] \right\} \geq 1-\epsilon,  \nonumber\\
\end{align}
because $|L_t^*|\leq Dt+1$ and function $f(|L_t^*|) =\left(1 - \frac{1}{|L_t^*|}\right) \sqrt{|L^*_t|}$ is increasing on $|L_t^*|$. \\ \\
From  Proposition 2 we have that given Assumptions 1 and 2 it holds that $V_q(L_t^*)<\lambda^{h(t)} {V_q}^t(Root)$ for all $t$ in $[1,2,...,T]$ with probability at least $1-2\epsilon$, where $h(t)$ is the height of leaf $L_t^*$ containing $x^*_t$. Thus, for the leaf containing $x^*_t$ we have $\sigma_{L^*_t}^2<\lambda^{h(t)} V^t_q(Root)\leq \frac{\lambda^{h(t)} \cdot t^2}{4}$, with probability at least $1-2\epsilon$, where the second inequality follows from Popoviciu's inequality on variances and the fact that Q-values are in the range $[0,t]$. We combine this result with inequality A.3, taking a lower bound over the two errors and we obtain:
\allowdisplaybreaks
\begin{align}
&\mathbb{P}\left\{J-J^* \geq  - \sum\limits_{t=1}^{T}\left[\left(1 - \frac{1}{Dt+1}\right) \sqrt{\frac{T(Dt+1)}{\epsilon}}\sqrt{\lambda^{h(t)} \cdot t^2} \right] \right\} \geq 1-3\epsilon \xRightarrow{} \tag{3.1}\\
&\mathbb{P}\left\{J-J^* \geq  - \frac{\sqrt{T}}{\sqrt{\epsilon}}\cdot\sum\limits_{t=1}^{T}\left[\left(\sqrt{Dt^3+t^2} \right){\sqrt{\lambda}}^{h(t)} \right] \right\} \geq 1-3\epsilon \tag{3.2}
\end{align}
The above calculations are based on the observation that for a random variable $X$ the probability $P\{X>a\}$ is greater than $P\{X>b\}$ for $a<b$. Thus, we can go from step to step by reducing the quantity that bounds $J-J^*$ from below. To go from step $(3.1)$ to $(3.2)$ we omit the term $- \frac{1}{Dt+1}$ in $\left(1 - \frac{1}{Dt+1}\right)$. 
\end{proof}
Next we give the proof of Theorem 3.a.\\
\noindent\textit{Statement} If Assumptions 1, 2 and 3.a hold, then $$\mathbb{P}\left\{J-J^* \geq  - \frac{\sqrt{T(D+1)}}{\sqrt{\epsilon\lambda}}\cdot \frac{\Bar{\lambda}(\Bar{\lambda} + 1)( \Bar{\lambda}^T-1)}{(\Bar{\lambda} - 1)^3}
\right\} \geq 1-3\epsilon$$ where $T$ is the total number of timesteps, $D$ is an upper bound for the number of new URLs inserted in frontier at each timestep, parameter $\lambda$ is introduced in Assumption 2 and $\Bar{\lambda}=\sqrt[2M_s]{\lambda}$.

\begin{proof}
Assumption 3.a implies for the height $h(t)$ of the optimal leaf $L_t^*$ that 
    $h(t) \geq floor(t / M_s)$, where $M_s$ is constant (independent of $T$). 
Starting from the result in Lemma 1 and using Assumption 3.a we obtain:
\allowdisplaybreaks
\begin{align}
&\mathbb{P}\left\{J-J^* \geq  - \frac{\sqrt{T}}{\sqrt{\epsilon}}\cdot\sum\limits_{t=1}^{T}\left[\left(\sqrt{Dt^3+t^2} \right){\sqrt{\lambda}}^{h(t)} \right] \right\} \geq 1-3\epsilon  \xRightarrow{} \tag{3.2}\\
&\mathbb{P}\left\{J-J^* \geq  - \frac{\sqrt{T(D+1)}}{\sqrt{\epsilon\lambda}}\cdot\sum\limits_{t=1}^{T}\left[t^{\frac{3}{2}}{\sqrt{\lambda}}^{\frac{t}{M_s}} \right] \right\} \geq 1-3\epsilon \xRightarrow[]{} \tag{3.3}\\
&\mathbb{P}\left\{J-J^* \geq  - \frac{\sqrt{T(D+1)}}{\sqrt{\epsilon\lambda}}\cdot\sum\limits_{t=1}^{T}t^2\cdot{\Bar{\lambda}}^t \right\} \geq 1-3\epsilon \Leftrightarrow \tag{3.4}\\
&\mathbb{P}\left\{J-J^* \geq  - \frac{\sqrt{T(D+1)}}{\sqrt{\epsilon\lambda}}\cdot \frac{\Bar{\lambda} [((\Bar{\lambda} - 1)^2 T^2 - 2 (\Bar{\lambda} - 1) T + \Bar{\lambda} + 1) \Bar{\lambda}^T - \Bar{\lambda} - 1]}{(\Bar{\lambda} - 1)^3}
\right\} \geq 1-3\epsilon \xRightarrow[]{}  \tag{3.5}\\
&\mathbb{P}\left\{J-J^* \geq  - \frac{\sqrt{T(D+1)}}{\sqrt{\epsilon\lambda}}\cdot \frac{\Bar{\lambda}(\Bar{\lambda} + 1)( \Bar{\lambda}^T-1)}{(\Bar{\lambda} - 1)^3}
\right\} \geq 1-3\epsilon \xRightarrow[]{}  \tag{3.6}\\
% &\mathbb{P}\left\{J-J^* \geq  - \frac{\sqrt{T}}{\sqrt{\epsilon\lambda}}\cdot\sum\limits_{t=1}^{T}\left[(Dt+1){\sqrt{\lambda}}^{\frac{t}{M_s}} \right] \right\} \geq 1-3\epsilon\xRightarrow[]{} \tag{3.3}\\
% &\mathbb{P}\left\{J-J^* \geq   \frac{\sqrt{T}}{\sqrt{\epsilon\lambda}}\cdot \frac{\Bar{\lambda}}{(1-\Bar{\lambda})^2}\cdot\left( 
% (D T + D + 1) {\Bar{\lambda}}^T - (D T + 1) \Bar{\lambda}^{(T + 1)} + \Bar{\lambda} - D - 1
% \right)\right\} \geq 1-3\epsilon\Leftrightarrow \tag{3.4}\\
% &\mathbb{P}\left\{J-J^* \geq   \frac{\sqrt{T}}{\sqrt{\epsilon\lambda}}\cdot \frac{\Bar{\lambda}}{(1-\Bar{\lambda})^2}\cdot\left[(\Bar{\lambda}^T-1)(D+1-\Bar{\lambda})+\Bar{\lambda}^T DT (1-\Bar{\lambda})\right] \right\} \geq 1-3\epsilon   \xRightarrow{} \tag{3.5} \\
% &\mathbb{P}\left\{J-J^* \geq   \frac{\sqrt{T}}{\sqrt{\epsilon\lambda}}\cdot \frac{\Bar{\lambda}}{(1-\Bar{\lambda})^2}\cdot(\Bar{\lambda}^T-1)(D+1-\Bar{\lambda}) \right\} \geq 1-3\epsilon  \tag{3.6} \\
\nonumber \end{align}
\\
The step from $(3.2)$ to $(3.3)$ follows from the inequality $h(t)\geq \frac{t}{M_s}-1$ and the observation that $Dt^3+t^2$ is greater than $(D+1)t^3$ . The transition from $(3.3)$ to $(3.4)$ can be made because $t^{\frac{3}{2}}<t^2$. Also, we set $\Bar{\lambda}:=\sqrt[2M_s]{\lambda}$. To go from $(3.4)$ to $(3.5)$ we use a (classical) formula for the sums $\sum_{k=1}^{n}k^2\cdot a^k$ and make an appropriate grouping of the terms, which will help us get rid of terms proportional to $T$ and $T^2$ and thus overcome the polynomial dependence of the error on $T$. The step from  $(3.5)$ to $(3.6)$ follows from the fact that $\Bar{\lambda}<1$, so the coefficients of $T^2$ and $T$ are positive.
\end{proof}
% \noindent We conclude that the absolute error is proportional to the square root of time horizon $T$. It is also proportional to the square root of $D$, where $D$ is the maximum number of new frontier samples at each timestep. Finally, there is a dependence on  $\frac{1}{\sqrt[2M_s]{\lambda}-1}$, where $\lambda$ is the decrease rate of the variance of Q-values along a path from the root to a leaf.

\noindent
Next we give the proof of Theorem 3.b.\\

\medskip

\noindent\textit{Statement}
If Assumptions 1, 2 and 3.b hold and parameter $\lambda$ introduced in Assumption 2 is less than $\frac{1}{8}$, then $\mathbb{P}\left\{J-J^* \geq  - \frac{T\sqrt{T(D+1)}}{\sqrt{\epsilon}} \right\} \geq 1-3\epsilon$, where $T$ is the total number of timesteps, $D$ is an upper bound for the number of new URLs inserted in frontier at each timestep and $\Bar{\lambda}=\sqrt[2M_s]{\lambda}$.
\begin{proof}
Starting from the result in Lemma 1 and using Assumption 3.b we obtain:\\
\allowdisplaybreaks
\begin{align}
&\mathbb{P}\left\{J-J^* \geq  - \frac{\sqrt{T}}{\sqrt{\epsilon}}\cdot\sum\limits_{t=1}^{T}\left[\left(\sqrt{Dt^3+t^2} \right){\sqrt{\lambda}}^{h(t)} \right] \right\} \geq 1-3\epsilon  \xRightarrow[(D+1)t^3\geq Dt^3+t^2]{h(t)\geq\log_2(t)} \tag{3.2}\\
&\mathbb{P}\left\{J-J^* \geq  - \frac{\sqrt{T(D+1)}}{\sqrt{\epsilon}}\cdot\sum\limits_{t=1}^{T}\left[t^{\frac{3}{2}}{\sqrt{\lambda}}^{\log_2(t)} \right] \right\} \geq 1-3\epsilon \Leftrightarrow \tag{3.3'}\\
&\mathbb{P}\left\{J-J^* \geq  - \frac{\sqrt{T(D+1)}}{\sqrt{\epsilon}}\cdot\sum\limits_{t=1}^{T}\left[t^{\frac{3}{2}}{t}^{\log_2(\sqrt{\lambda})} \right] \right\} \geq 1-3\epsilon \xRightarrow{\lambda < \frac{1}{8}} \tag{3.4'}\\
&\mathbb{P}\left\{J-J^* \geq  - \frac{\sqrt{T(D+1)}}{\sqrt{\epsilon}}\cdot\sum\limits_{t=1}^{T}\left[t^{\frac{3}{2}}{t}^{\log_2(\sqrt{\frac{1}{8}})} \right] \right\} \geq 1-3\epsilon \Leftrightarrow \tag{3.5'}\\
&\mathbb{P}\left\{J-J^* \geq  - \frac{\sqrt{T(D+1)}}{\sqrt{\epsilon}}\cdot\sum\limits_{t=1}^{T}[1] \right\} \geq 1-3\epsilon \Leftrightarrow \mathbb{P}\left\{J-J^* \geq  - \frac{T\sqrt{T(D+1)}}{\sqrt{\epsilon}} \right\} \geq 1-3\epsilon \tag{3.6'}\\
\nonumber \end{align}
\end{proof}

\noindent
Finally, we give the proof of Theorem 3.c.\\

\medskip

\noindent\textit{Statement} If Assumptions 1, 2 and 3.b hold and parameter $\lambda$ introduced in Assumption 2 is less than $\frac{1}{32}$, then $\mathbb{P}\left\{J-J^* \geq  - \frac{\sqrt{T(D+1)}}{\sqrt{\epsilon}}\cdot(ln(T)+1) \right\} \geq 1-3\epsilon$, where $T$ is the total number of timesteps, $D$ is an upper bound for the number of new URLs inserted in frontier at each timestep and 
\begin{proof}
Starting from the result in Lemma 1 and using Assumption 3.b we obtain:\\
\allowdisplaybreaks
\begin{align}
&\mathbb{P}\left\{J-J^* \geq  - \frac{\sqrt{T(D+1)}}{\sqrt{\epsilon}}\cdot\sum\limits_{t=1}^{T}\left[t^{\frac{3}{2}}{t}^{\log_2(\sqrt{\lambda})} \right] \right\} \geq 1-3\epsilon \xRightarrow{\lambda < \frac{1}{32}} \tag{3.4'}\\
&\mathbb{P}\left\{J-J^* \geq  - \frac{\sqrt{T(D+1)}}{\sqrt{\epsilon}}\cdot\sum\limits_{t=1}^{T}\left[t^{\frac{3}{2}}{t}^{\log_2(\sqrt{\frac{1}{32}})} \right] \right\} \geq 1-3\epsilon \Leftrightarrow \tag{3.5''}\\
&\mathbb{P}\left\{J-J^* \geq  - \frac{\sqrt{T(D+1)}}{\sqrt{\epsilon}}\cdot\sum\limits_{t=1}^{T}\frac{1}{t} \right\} \geq 1-3\epsilon \xRightarrow{H_n < ln(n)+1}\tag{3.6''}\\
&\mathbb{P}\left\{J-J^* \geq  - \frac{\sqrt{T(D+1)}}{\sqrt{\epsilon}}\cdot(ln(T)+1) \right\} \geq 1-3\epsilon \tag{3.7''}\\
\nonumber \end{align}
\end{proof}

\newpage

\section{Missing Algorithm}

    \begin{algorithm*}[h!]
    \caption{Tree REinforcement Spider (TRES)}\label{alg:3}
        \hspace*{\algorithmicindent} \textbf{Input:} Seed set $U_S$, initial keyword set $KS$, a dataset of web pages $D_{tr}$, the total number of URL fetches $T$ \\
        \hspace*{\algorithmicindent} \, \, \, \, \, \, \, the maximal number of domain visits $MAX$\\
        \hspace*{\algorithmicindent} \textbf{Output:} a list $\pmb{res}$ of the fetched URLs 
    \begin{algorithmic}[1]
        \State KS = KeywordExpansionStrategy($KS$, $D_{tr}$) \Comment{Algorithm 1}
        \State Learn the reward function $R$, by training a KwBiLSTM using $D_{tr}$ and $K$
        \State Initialize state $S_0 = g_0$ leveraging the experience from the relevant URLs in $U_S$ 
        \State Initialize experience samples $D_E(0)$ and frontier samples $D_F(0)$ using the seed set $U_S$
        \State Initialize closure $C_0$ and Tree-Frontier $F_0$ using $D_E(0)$ and $D_F(0)$
        \State Initialize timestep counter $t = 0$
        \State Initialize $e_{new} = \{\}$ and $f_{new} = \{\}$
        \State Initialize Experience Replay buffer $B$ with experience samples $D_E(0)$
        \State Initialize an empty list $\pmb{res}$
        \While{$t < T$}
            \State Train the DDQN agent on a minibatch of $B$ 
            % \State $\tau = t - 1$
            \State Get a state-action $\pmb{x}(s_{t},a_{t})$, $D_E(t+1)$ and $D_F(t+1)$ = TreeFrontierUpdate($F_{t}$, $D_E({t})$, $D_F({t})$, $e_{new}$, $f_{new}$) \textbf{s.t}:
            \State \, \, \, \, $a_{t} \rightarrow$  URL $u^{(t+1)} \not \in C_{t}$ \textbf{and} $MAX$ condition holds \Comment{Algorithm 3}
            
            \State Transition to a new state $S_{t+1} = g_{t+1}$ by expanding $g_t$ with the edge $a_{t}$
            \State Fetch URL $u^{(t+1)}$ and observe its reward $r_{t+1}$ 
            \State Extract outlink URLs $Outlink(u^{(t+1)})$ to create new frontier samples $f_{new}$
            \State Create new experience sample $e_{new} = (\pmb{x}(s_{t},a_{t}), r_{t+1})$
            \State Update experience replay buffer $B$ with $e_{new}$
            \State Update closure $C_{t+1}$ with the new fetched URL $u^{(t+1)}$
            \State Append $u^{(t+1)}$ to $\pmb{res}$
            \State $t = t + 1$
        \EndWhile
    \end{algorithmic}
    \end{algorithm*}

% \section{Harvest Rate Plots}

% \begin{figure}[h!]
%         \centering
%         \begin{subfigure}[t]{0.46\textwidth}
%             \includegraphics[width=\textwidth]{Results_of_Sport.png}
%             % \caption{Harvest Rate diagram for the Sports domain}
%         \end{subfigure}
%         \hfill
%         \begin{subfigure}[t]{0.46\textwidth}
%             \includegraphics[width=\textwidth]{Results_of_Food.png}
%             % \caption{Harvest Rate diagram for the Food domain}
%         \end{subfigure}\\
%         \begin{subfigure}[t]{0.46\textwidth}
%             \centering
%             \includegraphics[width=\textwidth]{Results_of_Hardware.png}
%             % \caption{Harvest Rate diagram for the Hardware domain}
%         \end{subfigure}\\
%         \caption{Harvest Rate diagrams: Sports (Up left), Food (Up right) and Hardware (Down)}
%         \label{fig:harvest}
% \end{figure}





% \end{thebibliography}
\end{document}